\documentclass[11pt,a4paper]{article}
\usepackage[utf8]{inputenc}
\usepackage[T1]{fontenc}

\usepackage[dvipsnames,usenames]{color}
\usepackage[usenames,dvipsnames]{xcolor}

\usepackage{amsthm,amsmath,amssymb,breqn}
\usepackage{fullpage}
\usepackage[ruled,noend,linesnumbered]{algorithm2e}     
\usepackage{enumerate,comment}
\usepackage{url}
\usepackage[capitalise]{cleveref}
\usepackage{todonotes}
\usepackage{tikz}
\usepackage[noadjust]{cite}
\usepackage{xspace}
\usepackage{graphicx}
\usepackage{float}
\usepackage{amsfonts}
\usepackage{color}
\usepackage{array}
\usepackage[colorlinks=true,urlcolor=Blue,citecolor=Green,linkcolor=BrickRed]{hyperref}
\newcounter{mycounter}  
\newenvironment{noindlist}
 {\begin{list}{\arabic{mycounter}.~~}{\usecounter{mycounter} \labelsep=0em \labelwidth=0em \leftmargin=0em \itemindent=0em}}
 {\end{list}}

\newcommand{\head}{\textsf{head}}
\newcommand{\tail}{\textsf{tail}}
\newcommand{\ceil}[1]{\left\lceil #1 \right\rceil}

\newcommand{\Oh}{{O}}

\newtheorem{theorem}{Theorem}[section]

\newtheorem{lemma}{Lemma}

\newtheorem{invariant}{Invariant}

\theoremstyle{definition}   

\usepackage{authblk}
\theoremstyle{remark}

\newcolumntype{R}[1]{>{\raggedleft\let\newline\\\arraybackslash\hspace{0pt}}m{#1}}

\begin{document}

\title{Dispersion on Trees\thanks{The research was supported in part by Israel Science Foundation grant 794/13.}}
\author[1]{Pawe\l{} Gawrychowski}
\author[1]{Nadav Krasnopolsky}
\author[2]{Shay Mozes}
\author[1]{Oren Weimann}
\affil[1]{University of Haifa, Israel}
\affil[2]{IDC Herzliya, Israel}

\date{}
\maketitle

\begin{abstract}
In the $k$-dispersion problem, we need to select $k$ nodes of a given graph so as to maximize
the minimum distance between any two chosen nodes. This can be seen as a generalization
of the independent set problem, where the goal is to select nodes so that the minimum distance
is larger than 1.
We design an optimal $\Oh(n)$ time algorithm for the dispersion problem on trees consisting
of $n$ nodes, thus improving the previous  $\Oh(n\log n)$ time solution from 1997. 

We also consider the weighted case, where the goal is to choose a set of nodes of total weight at least $W$. We present an $\Oh(n\log^2n)$ algorithm improving the previous $\Oh(n\log^4 n)$ solution. Our solution builds on the search version (where we know the minimum distance $\lambda$ between the chosen nodes) for which we present tight $\Theta(n\log n)$ upper and lower bounds. 
\end{abstract}

\section{Introduction}

\emph{Facility location} is a family of problems dealing with the placement of facilities on a network in order to optimize certain distances between the facilities, or between facilities and other nodes of the network. Such problems are usually if not always NP-hard on general graphs. There is a rich literature on approximation algorithms (see e.g.~\cite{DavidB.Shmoys1997,Vazirani2003} and references therein) as well as exact algorithms for restricted inputs. In particular, many linear and near-linear time algorithms were developed for facility location problems on edge-weighted trees.

In the most basic problem, called \emph{$k$-center}, we are given an edge-weighted tree with $n$ nodes and wish to designate up to $k$ nodes to be facilities, so as to minimize the maximum distance of a node to its closest facility. This problem was studied in the early 80's by Megiddo et al.~\cite{Megiddo1981} who gave an $\Oh(n\log^2n)$ time algorithm that was subsequently improved to $\Oh(n\log n)$ by Frederickson and Johnson \cite{Frederickson1983}.
In the early 90's, an optimal $\Oh(n)$ time solution was given by Frederickson~\cite{Frederickson1991a,Frederickson1990} using a seminal approach based on
parametric search, also for two other versions where points on edges can be designated as facilities or where we minimize over points on edges.
In yet another variant, called weighted $k$-center, every node has a positive weight and we wish to minimize the maximum weighted distance of a node to its closest facility. Megiddo et al.~\cite{Megiddo1981} solved this in $\Oh(n\log^{2}n)$ time, and
Megiddo and Tamir~\cite{Megiddo1983} designed an $\Oh(n\log^{2}n\log\log n)$ time algorithm when allowing points on edges to be
designated as facilities. The latter complexity can be further improved to $\Oh(n\log^{2}n)$ using a technique of Cole~\cite{Cole87}. 
A related problem, also suggested in the early 80's  \cite{Becker1982,Perl1981}, is \emph{$k$-partitioning}. In this problem the nodes have weight and we wish to delete $k$ edges in the tree so as to maximize the weight of the lightest resulting subtree. This problem was also solved by Frederickson in $\Oh(n)$ time \cite{Frederickson1991} using his parametric search framework. 

The focus of this paper is the  {\em $k$-dispersion} problem, where we wish to designate $k$ nodes as facilities so as to maximize the distances among the facilities. In other words, we wish to select $k$ nodes that are as spread-apart as possible. More formally, let $d(u,v)$ denote the distance between nodes $u$ and $v$, and for a subset of nodes $P$ let $f(P)=\min_{u,v\in P} \{d(u,v)\}$.

\begin{itemize} 
\item {\em The Dispersion Optimization Problem.} Given a tree with non-negative edge lengths, and a number $k$, find a subset $P$ of nodes of size $k$ such that $f(P)$ is maximized.
\end{itemize}

\noindent The dispersion problem can be seen as a generalization of the classical maximum independent set problem (that can be solved by binary searching for the largest value of $k$ for which the minimum distance is at least 2).  It can also be seen as a generalization of the diameter problem (i.e., when $k=2$).
It turns out that the dispersion and the $k$-partitioning problems are actually equivalent in the one-dimensional case (i.e., when the tree is a path). The reduction simply creates a new path whose edges correspond to nodes in the original path and whose nodes correspond to edges in the original path. However, such equivalence does not apply to general trees, on which $k$-dispersion seems more difficult than $k$-partitioning. In particular, until the present work, no linear time solution for $k$-dispersion was known. 
The dispersion optimization problem can be solved by repeatedly querying a {\em feasibility test} that solves the dispersion search problem. 

\begin{itemize}
\item  {\em The Dispersion Search Problem (feasibility test).}  Given a tree with non-negative edge lengths, a number $k$, and a number $\lambda$, find a subset $P$ of nodes of size $k$ such that  $f(P)\geq\lambda$, or declare that no such subset exists. 
\end{itemize}

\noindent Bhattacharya and Houle~\cite{Bhattacharya1991} presented a linear-time feasibility test, and used a result by Frederickson \cite{Frederickson1983} that enables binary searching over all possible values of $\lambda$ (i.e., all pairwise distances in the tree). That is, a feasibility test with a running time $\tau$ implies an $O(n \log n + \tau \cdot \log n)$ time algorithm for the dispersion optimization problem. Thus, the algorithm of Bhattacharya and Houle for the dispersion optimization problem runs in $O(n\log n)$ time. We present a linear time algorithm for the optimization problem. Our solution is based on a simplified linear-time feasibility test, which we turn into a sublinear-time feasibility test in a technically involved way closely inspired by Frederickson's approach.

In the {\em weighted} dispersion problem, nodes have non-negative weights. Instead of $k$ we are given $W$, and the goal is then to find a subset
$P$ of nodes of total weight at least $W$ s.t. $f(P)$ is maximized. Bhattacharya and Houle considered this generalization in~\cite{Bhattacharya1999}. They presented an $O(n \log^3 n)$ feasibility test for this generalization, that by the same reasoning above solves the weighted optimization problem in $O(n \log^4 n)$ time. We give an $O(n \log n$)-time feasibility test, and a matching lower bound.  Thus, our algorithm for the weighted optimization problem runs in $O(n \log^2 n)$ time. Our solution uses novel ideas, and differs substantially from Frederickson's approach.

\vspace{0.04in} \noindent {\bf Our technique for the unweighted dispersion problem.}
Our solution to the $k$-dispersion problem can be seen as a modern adaptation of Frederickson's approach based on a hierarchy of micro-macro decompositions. While achieving this adaptation is technically involved, we believe this modern view might be of independent interest. As in Frederickson's approach for $k$-partitioning and $k$-center, we develop a feasibility test that requires {\em linear} time preprocessing and can then be queried in \emph{sublinear} time.
Equipped with this sublinear feasibility test, it is still not clear how to solve the whole problem in $O(n)$ time, as in such complexity it is not trivial to represent all the pairwise distances in the tree in a structure that enables binary searching. To cope with this, we maintain only a subset of candidate distances and represent them using matrices where both rows and columns are sorted. Running feasibility tests on only a few candidate entries from such matrices allows us to eliminate many other candidates, and prune the tree accordingly. We then repeat the process with the new smaller tree.
This is similar to Frederickson's approach, but our algorithm (highlighted below) differs in how we construct these matrices,  in how we partition the input tree, and in how we prune it. 

Our algorithm begins by partitioning the input tree $T$ into $O(n/b)$ {\em fragments}, each with $O(b)$ nodes and at most two {\em boundary nodes} incident to nodes in other fragments: the root of the fragment and, possibly, another boundary node called the {\em hole}.
We use this to simulate a bottom-up feasibility test by jumping over entire 
fragments, i.e., knowing $\lambda$, we wish to extend in $O(\log b)$ time a solution for a subtree of $T$ rooted at the fragment's hole to a subtree of $T$ rooted at the fragment's root. This is achieved by efficient preprocessing: 
The first step of the preprocessing computes values $\lambda_1$ and  $\lambda_2$ such that (1) there is no solution to the search problem on $T$ for any $\lambda \geq \lambda_2$, (2) there is a solution to the search problem on $T$ for any $\lambda \le \lambda_1$, and (3) for {\em most} of the fragments, the distance between any two nodes is either smaller or equal to $\lambda_1$ or larger or equal to $\lambda_2$. This is achieved by applying Frederickson's parametric search on sorted matrices capturing the pairwise distances between nodes in the same fragment. The (few) fragments that do not satisfy property (3) are handled naively in $O(b)$ time during query time. 
The fragments that do satisfy property (3) are further preprocessed. We look at the path from the hole to the root of the
fragment and run the linear-time feasibility test for all subtrees hanging off from it. Because of property (3), this can be done in advance without knowing the actual exact value of $\lambda \in (\lambda_1,\lambda_2)$, which will only be determined at query time. 
Let $P$ be a solution produced by the feasibility test to a subtree rooted at a node $u$. It turns out that the interaction between $P$ and the solution to the entire tree depends only on two nodes of $P$, which we call the {\em certain} node and the {\em candidate} node. We can therefore conceptually replace each hanging subtree by two leafs, and think of the fragment as a caterpillar connecting the root and the hole. 
After some additional pruning, we can precompute information 
that will be used to accelerate queries to the feasibility test. During a query we will be able to jump over each fragment of size $O(b)$ in just $O(\log b)$ time, so the test takes $O(\frac{n}{b}\log b)$ time.

The above sublinear-time feasibility test is presented in~Section~\ref{sublinear f.t.}, with an overall preprocessing time of $O(n\log\log n)$.
The test is then used to solve the optimization problem within the same time. This is done, again, by maintaining an interval
$[\lambda_{1},\lambda_{2})$ and applying Frederickson's parametric search, but now we apply a heavy path decomposition to construct
the sorted matrices.
To accelerate the $O(n\log\log n)$ time algorithm, we construct a hierarchy of feasibility tests by partitioning the input tree
into larger and larger fragments. In each iteration we construct a feasibility test with better running time, until finally, after $\log^*n$ iterations
we obtain a feasibility test with $O(\frac{n}{\log ^4n} \cdot \log \log n)$ query-time, which we use to solve the dispersion optimization problem in linear  time. It is relatively straightforward to implement the precomputation done in a single iteration in $O(n)$ time. However, achieving total $O(n)$ time over all the iterations, requires reusing the results of the precomputation across iterations as well as an intricate global analysis of the overall complexity.
The algorithm is described in Section \ref{linear algorithm for the dispersion optimization problem}.

\vspace{0.04in} \noindent {\bf Our technique for the weighted dispersion problem.}
Our solution for the weighted case differs substantially from Frederickson's approach.
In contrast to the unweighted case, where it suffices to consider a single candidate node, in the weighted case each subtree might have a large number of candidate nodes. 
  To overcome this, we represent the candidates of a subtree with a {\em monotonically decreasing polyline}:
for every possible distance $d$, we store the maximum weight $W(P)$ of a subset of nodes $P$ such that the distance of every node of $P$ to the root of the subtree is at least $d$.
This can be conveniently represented by a sorted list of breakpoints, and the number of breakpoints is at most
the size of the subtree. We then show that the polyline of a node can be efficiently computed by merging the polylines of its children. If the polylines
are stored in augmented balanced search trees, then two polylines of size $x$ and $y$ can be merged in time $\Oh(\min(x,y)\log\max(x,y))$,
and by standard calculation we obtain an $\Oh(n\log^{2}n)$ time feasibility test. To improve on that and obtain an optimal $\Oh(n\log n)$ feasibility test,
we need to be able to merge polylines in $\Oh(\min(x,y)\log\frac{\max(x,y)}{\min(x,y)})$ time. 
An old result of Brown and Tarjan~\cite{Brown1980} is that, in exactly such time we can merge two {\em 2-3 trees} representing two sorted lists of length $x$ and $y$ (and also delete $x$ nodes in a tree of size $y$). This was later generalized by Huddleston and
Mehlhorn~\cite{huddlestonM82} to any sequence of operations that exhibits a certain locality of reference. However, in our specific
application we need various non-standard batch operations on the lists.
We present a simpler data structure for merging polylines that efficiently supports the required batch operations and works with  
any balanced search tree with split and join capabilities. Our data structure both simplifies and extends that of Brown and Tarjan~\cite{Brown1980}, and we believe it to be of independent interest.

\section{A Linear Time Feasibility Test}
\label{linear F.T.}

Given a tree $T$ with non-negative lengths and a number $\lambda$, the feasibility test finds a subset of nodes $P$ such that $f(P)\geq\lambda$
and $|P|$ is maximized, and then checks if $|P|\geq k$.
To this end, the tree is processed bottom-up while computing, for every subtree $T_{r}$ rooted at a node $r$, a subset of nodes $P$ such that
$f(P)\geq\lambda$, $|P|$ is maximized, and in case of a tie $\min_{u\in P}d(r,u)$ is additionally maximized.
We call the node $u\in P$, s.t. $d(r,u)<\frac{\lambda}{2}$, the \emph{candidate} node of the subtree (or a candidate with respect to $r$). There is at most one such candidate node.
The remaining nodes in $P$ are called \emph{certain} (with respect to $r$) and the one that is nearest to the root is called the certain node.
When clear from the context, we will not explicitly say which subtree we are referring to.

In each step we are given a node $r$, its children nodes $r_{1},r_{2},\ldots,r_{\ell}$ and, for each child $r_{i}$, a maximal valid
solution $P_{i}$ for the feasibility test on $T_{r_{i}}$ together with the candidate and the certain node. We obtain a maximal valid solution $P$ for the feasibility test on $T_{r}$ as follows:
\begin{enumerate}
\item Take all nodes in $P_{1},\ldots,P_{\ell}$, except for the candidate nodes.\label{linear time step 1}
\item Take all candidate nodes $u$ s.t. $d(u,r) \geq \frac{\lambda}{2}$ (i.e., they are certain w.r.t. $r$).\label{linear time step 2}
\item If it exists, take $u'$, the candidate node farthest from $r$ s.t. $d(u',r) < \frac{\lambda}{2}$ and $d(u',x)\geq \lambda$, where $x$ is the closest node to $u'$ we have taken so far.\label{linear time step 3}
\item If the distance from $r$ to the closest vertex in $P$ is at least $\lambda$, add $r$ to $P$.\label{linear time step 4}
\end{enumerate}
Iterating over the input tree bottom-up as described results in a valid solution $P$ for the whole tree. Finally, we check if $|P|\geq k$.

\begin{lemma}
\label{lineartimecorrectness}
The above feasibility test works in linear time and finds $P$ such that $f(P)\geq\lambda$ and $|P|$ is maximized.
\end{lemma}

\begin{proof}
We first analyze the time complexity. To show that the total running time is $O(n)$, we only have to argue that each step of the bottom up computation takes $O(\ell)$ time. In order to perform the algorithm efficiently we preprocess the input tree in linear time, and store the distance of every node to the root of the whole tree. We can now use lowest common ancestor queries~\cite{Bender2000}, to compute any pairwise distance in the tree in $O(1)$ time.

We now show that each step takes $O(\ell)$ time. Step \ref{linear time step 1} can be done in $O(\ell)$ time by storing the subsets computed in each iteration in any data structure that allows merging in constant time. Step \ref{linear time step 2} takes $O(\ell)$ time, since there is at most one candidate node in each subtree rooted at a child of $r$. In step \ref{linear time step 3} we find $x$, the node nearest to $r$ taken thus far, by iterating over the certain node of each of the subsets $P_{1},\ldots,P_{\ell}$, and the candidate nodes we have chosen in step \ref{linear time step 2}. The number of the nodes we need to consider is $O(\ell)$. We then iterate over the remaining candidate nodes, find the one farthest from $r$, and check if its distance to $x$ is greater or equal to $\lambda$. This is all done in $O(\ell)$ time. In Step \ref{linear time step 4} we only need to consider  the node $x$ which we found in Step \ref{linear time step 3}, since if the a candidate node w.r.t $r$ was taken in Step \ref{linear time step 3}, $r$ cannot be included in the solution.

To prove correctness, we will argue by induction that the following invariant holds: for every node $r$ of the tree, the algorithm
produces $P$, a subset of the vertices of the subtree rooted at $r$, s.t. $f(P)\geq\lambda$, $|P|$ is maximal and, if there are multiple
such valid subsets with maximal cardinality, $\min_{u\in P} d(r,u)$ is maximal. 

Consider a node $r$ of the tree and and assume that the invariant holds for all of its children $r_{1},r_{2},\ldots,r_{\ell}$.
We want to show that the invariant also holds for $r$.
Assume for contradiction that $P'$ is a subset of the vertices of $T_r$, s.t. $f(P')\geq\lambda$ and $|P'| > |P|$. Let us look at the two possible cases:
\begin{enumerate}
\item \textbf{\boldmath$P'$ has more certain nodes w.r.t. \boldmath$r$ than \boldmath$P$}: All certain nodes w.r.t. $r$ are in subtrees rooted at children of $r$, and so there must be some child $r_{i}$ of $r$ s.t. $P'_{r_i}$ has more certain nodes w.r.t. $r$ than $P_{r_i}$ (where $P'_{r_i}$ is $P'$ restricted to $T_{r_i}$). This can only be if $|P'_{r_i}| > |P_{r_i}|$, or if the closest node to $r_i$ in $P'_{r_i}$ is farther than the closest node in $P_{r_i}$, and so the invariant does not hold for $r_i$.
\item \textbf{\boldmath$P'$ has the same number of certain nodes w.r.t. \boldmath$r$ as \boldmath$P$, but \boldmath$P'$ has a candidate node w.r.t. \boldmath$r$  and \boldmath$P$ does not}: Denote the candidate node of $P'$ by $u$. We have two possible cases. First, if $u=r$, then there is some node $v \in P$ s.t. $d(u,v)<\lambda$. Assume that $v$ is in the subtree rooted at $r_i$. In this case, either $|P'_{r_i}| \geq |P_{r_i}|$ and the closest node to $r_i$ in $P'_{r_i}$ is farther than the closest node in $P_{r_i}$, and the invariant does not hold for $r_i$, or $|P'_{r_i}|<|P_{r_i}|$ and so there must be some $r_j$ s.t. $|P'_{r_j}|>|P_{r_j}|$, and the invariant does not hold for $r_j$. Second, if $u$ is in one of the subtrees rooted at children of $r$, since all certain nodes w.r.t. $r$ are also in these subtrees, there must be one such subtree, $T_{r_m}$ s.t. $|P'_{r_m}| > |P_{r_m}|$, and so the invariant does not hold for it.
\end{enumerate} 
Thus we have proven that $P$ is of maximal cardinality. Now, assume for contradiction that $P'$ is a subset of the vertices of the subtree rooted at $r$, s.t. $f(P')\geq\lambda$ and $|P'| = |P|$, but the closest node to $r$ in $P$ (denoted by $u$) is closer to $r$ than the closest node to $r$ in $P'$. Assume that $u \neq r$ and denote by $T_{r_i}$ the subtree that $u$ is in. Thus, either $|P'_{r_i}| = |P_{r_i}|$, but the closest vertex to $r_i$ in $P'_{r_i}$ is farther away from $r_i$ than $u$, which means that the invariant does not hold for $T_{r_i}$, or $|P'_{r_i}| < |P_{r_i}|$, and so there exists a subtree $T_{r_j}$, s.t. $|P'_{r_i}| > |P_{r_i}|$ and therefore the invariant does not hold for $r_j$. If $u=r$, then since $|P'| = |P|$, there must some child $r_{i}$ of $r$ s.t. $|P'_{r_i}| > |P_{r_i}|$, and so the invariant does not hold for $r_i$. We have proven that the invariant holds also for $r$ and, consequently, the algorithm is indeed correct.
\end{proof}

\section{An \texorpdfstring{\boldmath$O(n\log\log n)$}{O(nloglogn)} Time Algorithm for the Dispersion Problem}
\label{sublinear f.t.}

To accelerate the linear-time feasibility test described in Section~\ref{linear F.T.}, we will partition the tree into $O(n/b)$ {\em fragments}, each of size at most
$b$. We will preprocess each fragment s.t. we can implement the bottom-up feasibility test in sublinear time by  ``jumping'' over fragments in $O(\log b)$ time instead of $O(b)$.  The preprocessing takes $O(n\log b)$ time (Section~\ref{Pre-Processing Fragments}), and each feasibility test can then be implemented in sublinear $O(\frac{n}{b} \cdot \log b)$ time (Section~\ref{sec:feasibility test}). Using heavy-path decomposition, we design an algorithm for the unweighted dispersion optimization problem whose running time is dominated by the $O(\log^{2}n)$ calls it makes to the sublinear feasibility test  (Section~\ref{sec:nloglogn}). By setting $b=\log^{2}n$ we obtain an $O(n\log\log n)$ time algorithm.

\subsection{Tree partitioning}
We would like to partition the tree into fragments as follows. Each fragment is defined by one or two {\em boundary} nodes: a node $u$, and possibly a descendant of $u$, $v$. The fragment whose boundary nodes are $u$ and $v$ consists of the subtree of $u$ without the subtree of $v$ ($v$ does not belong to the fragment). Thus, each fragment is connected to the rest of the tree only by its boundary nodes. See Figure~\ref{fig:fragment}. 
We call the path from $u$ to $v$ the fragment's \textit{spine}, and $v$'s subtree its \textit{hole}. If the
fragment has only one boundary node, i.e., the fragment consists of a node and all its descendants, we say that there is no hole.
A partition of a tree into $O(n/b)$ such fragments, each of size at most $b$, is called a \emph{good partition}.
This is very similar to the definition of a cluster in top trees~\cite{TopTrees}, which are in turn inspired by topology trees~\cite{Frederickson1985},
but slightly tweaked to our needs.
Note that we can assume that the input tree is binary: given a non-binary tree, we can replace every node of degree $d\geq 3$ with a binary tree on
$d$ leaves. The edges of the binary tree are all of length zero, so at most one node in the tree can be taken.

\begin{figure}[ht]
\begin{center}
\includegraphics[scale=0.7]{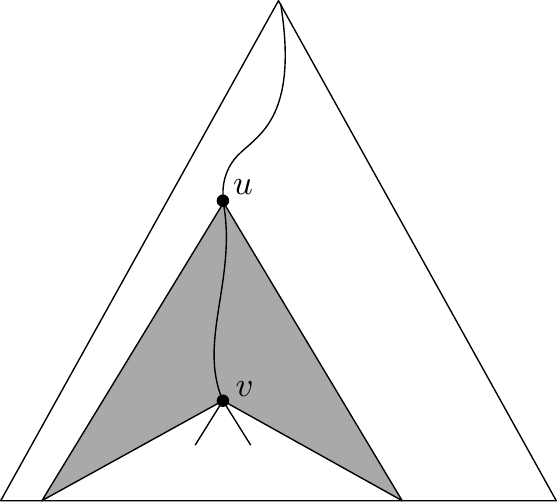}
\end{center}
\caption{A fragment in a good partition consists of the subtree of $u$ without the subtree of $v$.
$u$ is the root, $v$ is the hole, and the path from $u$ to $v$ is the spine of the fragment. \label{fig:fragment}}
\end{figure}

\begin{lemma}
\label{basic partitioning lemma}
For any binary tree on $n$ nodes and a parameter $b$, a good partition of the tree can be found in $O(n)$ time.
\end{lemma}

\begin{proof}
We prove the lemma by construction. Call a node \textit{large} if the size of its subtree is at least $b$
and \textit{small} otherwise. Consider the tree $T_L$ induced by the large nodes of the original tree. 
For each leaf $u\in T_{L}$, we make the subtree rooted at each of its children in the original tree a new fragment with no holes.
Each leaf of leaf of $T_{L}$ is the root of a subtree of size at least $b$ in the original tree, and these subtrees
are all disjoint, so we have at most $n/b$ leaves in $T_{L}$. Each of them creates up to two fragments
since the tree is binary, and each of these fragments is of size at most $b$ by definition.
The natural next step would be to keep cutting off maximal fragments of size $b$ from the remaining
part of the tree. This does not quite work, as we might create fragments with more than two boundary
nodes with such a method.
Therefore, the next step is to consider every branching node in $T_{L}$ instead, and make it a fragment
consisting of just one node. This also creates up to $n/b$ fragments, since in any tree the number of
branching nodes is at most the number of leaves.
Ignoring the already created fragments, we are left with large nodes that form unary chains in $T_{L}$.
Each of these nodes might also have an off-chain child that is a small node. We have $O(n/b)$ of these chains
(because each of them corresponds to an edge of a binary tree on at most $n/b$ leaves).
We scan each of these chains bottom-up and greedily cut them into fragments of size at most $b$.
Denoting the size of the $i$-th chain by $b_{i}$, and the number of chains by $k$, the number of fragments created in this phase is
bounded by $$\sum_{i=1}^{k} \left\lceil \frac{b_i}{b} \right\rceil \leq n/b+k=O(n/b).$$
In total we have created $O(n/b)$ fragments, and each of them is of size at most $b$.
The whole construction can be implemented by scanning the tree twice bottom-up in $O(n)$ time.
\end{proof}

\subsection{The preprocessing} \label{Pre-Processing Fragments}

Recall that the goal in the optimization problem is to find the largest feasible $\lambda^{*}$. Such $\lambda^{*}$
is a distance between an unknown pair of vertices in the tree. The first goal of the preprocessing step is to eliminate many possible pairwise distances, so that we can identify a small interval $[\lambda_1,\lambda_2)$ that contains $\lambda^*$. We want this interval to be sufficiently small so that for (almost) every fragment $F$, handling $F$ during the bottom up feasibility test for any value $\lambda$ in $[\lambda_1,\lambda_2)$ is the same. Observe that the feasibility test in Section~\ref{linear F.T.} for value $\lambda$ only compares distances to $\lambda$ and to $\lambda/2$. We therefore call a fragment $F$ \emph{inactive} if for any two nodes $u_1,u_2\in F$ the following two conditions hold: (1) $d(u_{1},u_{2})\leq \lambda_{1}$ or $d(u_{1},u_{2})\geq\lambda_{2}$, and (2) $d(u_{1},u_{2})\leq \frac{\lambda_{1}}{2}$ or $d(u_{1},u_{2})\geq \frac{\lambda_{2}}{2}$. For an inactive fragment $F$, all the comparisons performed by the feasibility test for any $\lambda \in (\lambda_1,\lambda_2)$ do not depend on the particular value of $\lambda$ in the interval. Therefore, once we find an interval $[\lambda_1,\lambda_2)$ for which (almost) all fragments are inactive, we can precompute, for each inactive fragment $F$, information that will enable us to process $F$ in $O(\log b)$ time during any subsequent feasibility test with $\lambda \in (\lambda_1,\lambda_2)$.

The first goal of the preprocessing step is therefore to find a small enough interval $[\lambda_1,\lambda_2)$. We call a matrix sorted if the elements of every row and every column are sorted (s.t. all the rows are monotone increasing or all the rows are monotone decreasing, and the same holds for the columns).
For each fragment $F$, we construct an implicit representation of $O(b)$ sorted matrices of total side length $O(b\log b)$ s.t. for every two nodes $u_1,u_2$ in $F$, $d(u_1,u_2)$ (and also $2d(u_1,u_2)$) is an entry in some matrix. This is done by using the following lemma. 
\begin{lemma}
Given a tree $T$ on $b$ nodes, we can construct in $O(b\log b)$ time an implicit representation of $O(b)$ sorted matrices of total
side length $O(b\log b)$ such that, for any $u,v\in T$, $d(u,v)$ is an entry in some matrix.
\end{lemma}
\begin{proof}
To construct the matrices we apply the standard centroid decomposition. A node $v\in T$ is a centroid if every connected component of
$T\setminus\{v\}$ consists of at most $\frac{|T|}{2}$ nodes. The centroid decomposition of $T$ is defined recursively by first choosing a centroid
$v\in T$ and then recursing on every connected component of $T\setminus\{v\}$, which we call \emph{pieces}.
Overall, this process takes $O(b\log b)$ time. Since we assume that the input tree is binary, the centroid in every recursive call splits the tree
into three pieces. Now, we run the following bottom-up computation: assuming that we are given the three pieces, and for each of
them a sorted list of the distances to their centroid, we would like to compute a sorted list of distances of the nodes in all three pieces
to the centroid used to separate them. The difficulty is that we would like to produce this sorted list in linear time.

\begin{figure}[ht]
\begin{center}
\includegraphics[scale=1]{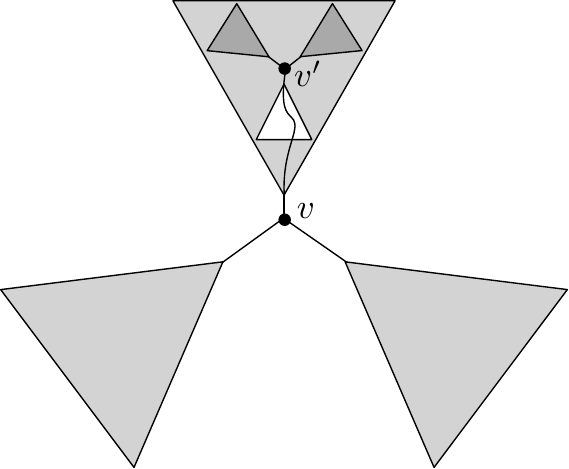}
\end{center}
\caption{A single step in the centroid decomposition. Removing $v$ splits $T$ into three pieces, s.t. none of them has more than $\frac{|T|}{2}$ nodes. The white piece is problematic, since the paths from its nodes to $v$ do not pass through the inner centroid $v'$. \label{fig:problematic}}
\end{figure}

Consider one of the three already processed pieces obtained by removing the outer centroid. It contains three smaller pieces of its own,
that have been obtained by removing the inner centroid.
For two of these smaller pieces, it holds that the path from every node inside to the outer centroid passes through the inner centroid.
We call the third piece, for which this does not hold, a \textit{problematic} piece. See Figure~\ref{fig:problematic}. For the two pieces that are not problematic, we can 
increase all entries of their lists by the distance from the inner centroid to the outer centroid, and then merge two sorted lists in linear time.
Now consider the problematic piece, where paths to the outer centroid do not go through the inner centroid. Notice that if we go one step deeper
in the recursion, the problematic piece consists of two pieces that are not problematic, as the paths from all their nodes to the outer centroid
pass through the inner centroid, and the remaining even smaller problematic piece. For the two non-problematic pieces we can again
increase all entries of their lists and then merge two sorted lists in linear time, and repeat the reasoning on the problematic piece.
If the initial problematic piece was of size $s$, after $O(\log s)$ iterations we obtain that a sorted list of distances to the centroid can be
obtained by merging sorted lists of length $s,s/2,s/4,\ldots$, which can be done in $O(s)$ time. At every level of recursion, 
the total size of all the pieces is at most $b$, and there are $O(\log b)$ levels of recursion, so the total time spend on constructing
the sorted lists is $O(b\log b)$.

Finally, for every piece we define a sorted matrix, that is, a matrix where every row is sorted and every column is sorted. If the sorted list of all distances of the nodes in the piece to the
centroid is $L[1..s]$, the entries of the matrix are $M[i,j]=L[i]+L[j]$. To represent $M$, we only have to store $L$, and then are able to
retrieve any $M[i,j]$ in $O(1)$ time. Note that some entries of $M$ do not really represent distances between two nodes, but this is irrelevant.
To bound the number of matrices, observe that the total number of pieces is $b$, because after constructing a piece we remove the centroid.
To bound the total side length, observe that the side length of $M$ is bounded by the size of the corresponding piece, and the total size
of all pieces at the same level of recursion is at most $b$, so indeed the total side length is $O(b\log b)$.
\end{proof}

Having constructed the representation of these matrices, we repeatedly choose an entry of a matrix and run a feasibility test with its value.   
Depending on the outcome, we then appropriately shrink the current interval $[\lambda_{1},\lambda_{2})$ and discard this entry.
Because the matrices are sorted, running a single feasibility test can actually allow us to discard multiple entries in the same matrix
(and, possibly, also entries in some other matrices). The following theorem by Frederickson shows how to exploit this to discard most
of the entries with very few feasibility tests.

\begin{theorem}[\cite{Frederickson1991}]\label{Frederickson's theorem}
Let  ${M_1, M_2, . . . , M_N}$ be a collection of sorted matrices in which matrix $M_j$ is of dimension $m_j \times n_j$, $m_j \leq n_j$, and $\sum_{j=1}^{N} m_j = m$.
Let $p$ be nonnegative. The number of feasibility tests needed to discard all but at most $p$ of the elements is $O(\max \lbrace \log(\max_{j} \lbrace n_j \rbrace), \log(\frac{m}{p+1}) \rbrace)$, and the total running time exclusive of the feasibility tests is $O(\sum_{j=1}^{N} m_j \cdot \log (2n_j/m_j))$.
\end{theorem}

Setting $m=b \log b \cdot \frac{n}{b} = n \log b$ and $p=n/b^2$, the theorem implies that we can use $O(\log b)$ calls to the linear time feasibility test and discard all but $n/b^2$ elements of the matrices. Therefore, all but at most $n/b^{2}$ fragments are inactive.

The second goal of the preprocessing step is to
compute information for each inactive fragment that will allow us to later ``jump'' over it in $O(\log b)$ time when running the feasibility test.  We next describe this computation. All future queries to the feasibility test will be given some number in the interval $(\lambda_1,\lambda_2)$ as the parameter $\lambda$ (since for any other value of $\lambda$ we already know the answer). For now we choose $\lambda$ arbitrarily in
$(\lambda_1,\lambda_2)$. This is done just so that we have a concrete value of $\lambda$ to work with.
\begin{noindlist}
\item\textbf{Reduce the fragment to a caterpillar:}
a fragment consists of the spine and the subtrees hanging off the spine. 
We run our linear-time feasibility test on the subtrees hanging off the spine, and obtain
the candidate and the certain node for each of them. The fragment can now be reduced to a caterpillar with at most two leaves attached to each
spine node: a candidate node and a certain node.
\item\label{removing certain nodes}
\textbf{Find candidate nodes that cannot be taken into the solution:}
for each candidate node we find its nearest certain node. Then, we compare their distance to $\lambda$ and remove the
candidate node if it cannot be taken. To find the nearest certain node, we first scan all nodes bottom-up (according to the natural order
on the spine nodes they are attached to) and compute for each of them the nearest certain node below it. Then, we repeat the scan
in the other direction to compute the nearest certain node above. This gives us, for every candidate node, the nearest certain node above
and below. We delete all candidate nodes for which one of these distances is smaller than $\lambda$.
We store the certain node nearest to the root, the certain node nearest to the hole and the total number of certain nodes,
and from now on ignore certain nodes and consider only the remaining candidate nodes.
\item\label{making distances from the root monotone}
\textbf{Prune leaves to make their distances to the root non-decreasing:}
let the $i$-th leaf, $u_{i}$, be connected with an edge of length $y_{i}$ to a spine node at distance $x_{i}$ from the root,
and order the leaves so that $x_{1}<x_{2}<\ldots<x_{s}$. See Figure~\ref{fig:pruning}.
Note that $y_{i}<\frac{\lambda}{2}$, as otherwise $u_{i}$ would be a certain node.
Suppose that $u_{i-1}$ is farther from the root than $u_i$ (i.e.,
$x_{i-1}+y_{i-1} > x_i+y_i$), then:
$d(u_{i},u_{i-1}) = x_i-x_{i-1}+y_i+y_{i-1} = x_{i} + y_{i} - x_{i-1} + y_{i-1} < 2y_{i-1} < \lambda.$
Therefore any valid solution cannot contain both $u_{i}$ and $u_{i-1}$. We claim that if there is an optimal solution which contains
$u_{i}$, replacing $u_{i}$ with $u_{i-1}$ also produces an optimal solution. To prove this, it is enough to argue that
$u_{i-1}$ is farther away from any node above it than $u_i$, and $u_i$ is closer to any node below it than $u_{i-1}$.
Consider a node $u_{j}$ that is above $u_{i-1}$ (so $j<i-1$), then:
$d(u_j,u_{i-1}) - d(u_j,u_{i}) = y_{i-1}-(x_i-x_{i-1})-y_i = x_{i-1}+y_{i-1}-(x_i+y_i) > 0.$
Now consider a node $u_{j}$ that is below $u_{i}$ (so $j>i$), then:
$d(u_j,u_{i-1}) - d(u_j,u_{i}) = y_{i-1}+(x_i-x_{i-1})-y_i > 2(x_i-x_{i-1}) > 0.$
So in fact, we can remove the $i$-th leaf from the caterpillar if $x_{i-1}+y_{i-1} > x_i+y_i$.
To check this condition efficiently, we scan the caterpillar from top to bottom while maintaining the most recently processed non-removed leaf.
This takes linear time in the number of candidate nodes and ensures that the distances of the
remaining leaves from the root are non-decreasing.
\item\label{making distances from the hole monotone}
\textbf{Prune leaves to make their distances to the hole non-increasing:}
this is done as in the previous step, except we scan in the other direction.
\item\label{precompute for any candidate node}
\textbf{Preprocess for any candidate and certain node with respect to the hole:}
we call $u_{1},u_{2},\ldots,u_{i}$ a {\em prefix} of the caterpillar and, similarly, $u_{i+1},u_{i+2}, \ldots,u_{s}$ a {\em suffix}.
For every possible prefix, we would like to precompute the result of running the linear-time feasibility
test on that prefix. In Section~\ref{sec:feasibility test} we will show that, in fact, this is enough to efficiently
simulate running the feasibility test on the whole subtree rooted at $r$ if we know the candidate node and the certain node w.r.t. the hole. Consider running the feasibility test on $u_{1},u_{2},\ldots,u_{i}$.
Recall that its goal is to choose as many nodes as possible, and in case of a tie to maximize the distance of the
nearest chosen node to $r$. Due to distances of the leaves to $r$ being non-decreasing, it is clear that
$u_{i}$ should be chosen. Then, consider the largest $i'<i$ such that $d(u_{i'},u_{i})\geq \lambda$.
Due to distances of the leaves to the hole
being non-decreasing, nodes $u_{i'+1},u_{i'+2},\ldots,u_{i-1}$ cannot be chosen and furthermore $d(u_{j},u_{i})\geq \lambda$
for any $j=1,2,\ldots,i'$. Therefore, to continue the simulation we should repeat the reasoning for $u_{1},u_{2},\ldots,u_{i'}$.
This suggests the following implementation: scan the caterpillar from top to bottom and store, for every prefix $u_{1},u_{2},\ldots,u_{i}$,
the number of chosen nodes, the certain node and the candidate node. While scanning we maintain $i'$ in amortized constant time.
After increasing $i$, we only have to keep increasing $i'$ as long as $d(u_{i},u_{i'}) \ge \lambda$.
To store the information for the current prefix, copy the computed information for $u_{1},u_{2},\ldots,u_{i'}$
and increase the number of chosen nodes by one. Then, if the certain node is set to NULL, we set it to be $u_{i}$. If there is no $u_{i'}$, and $u_i$ is the top-most chosen candidate, we need to set it to be the candidate (if $d(r,u_i) < \frac{\lambda}{2}$) or the certain node otherwise.
\end{noindlist}

\begin{figure}[ht]
\begin{center}
\includegraphics[scale=0.45]{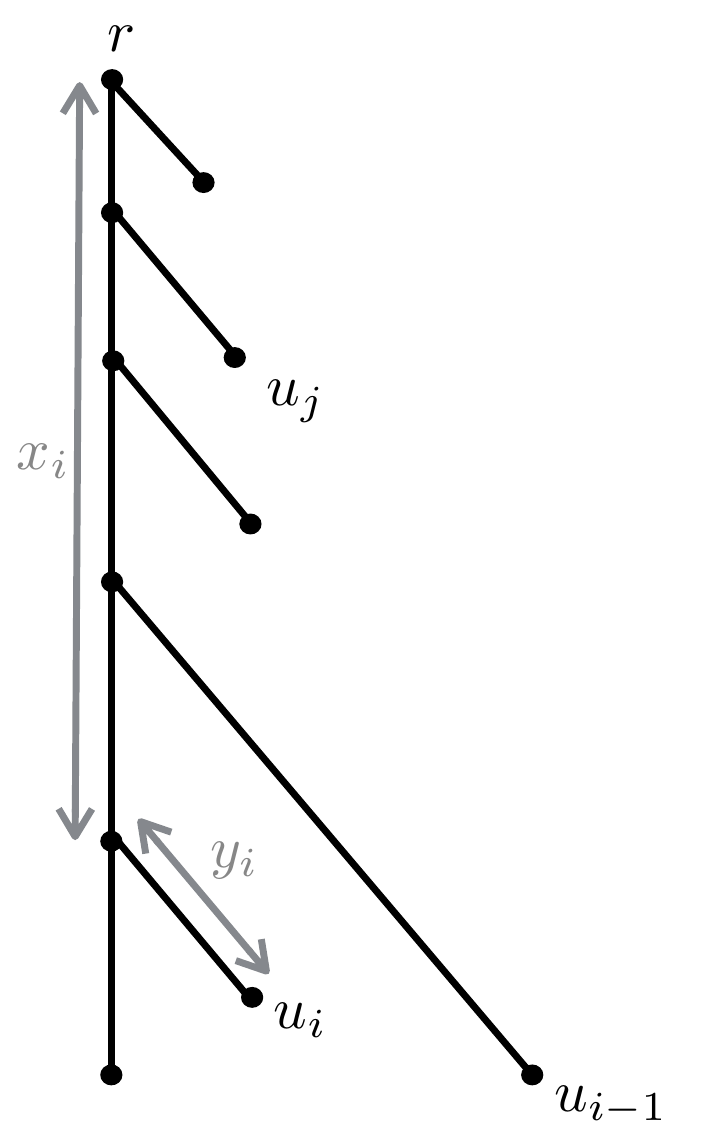}
\end{center}
\caption{If $d(u_i,r)<d(u_{i-1},r)$ then  we can remove $u_{i}$ since (1) an optimal solution cannot contain both $u_{i}$ and $u_{i-1}$ (the distance between them is too small), and (2) if an optimal solution contains $u_{i}$ then it can be replaced with $u_{i-1}$ (every leaf $u_j$ is farther from $u_{i-1}$ than from $u_i$). \label{fig:pruning}}
\end{figure}

\subsection{The feasibility test}
\label{sec:feasibility test}

The sublinear feasibility test for a value $\lambda \in (\lambda_1,\lambda_2)$ processes the tree bottom-up. For every fragment with root $r$, we would like to simulate running the
linear-time feasibility test on the subtree rooted at $r$ to compute: the number of chosen nodes, the candidate node, and the certain
node. We assume that we already have such information for the fragment rooted at the hole of the current fragment.
If the current fragment is active, we process it naively in $O(b)$ time
using the linear-time feasibility test. If it is inactive, we process it (jump over it) in $O(\log b)$ time.
This can be seen as, roughly speaking, attaching the hole as another spine node to the corresponding caterpillar and executing steps
(\ref{removing certain nodes})-(\ref{precompute for any candidate node}).

We start by considering the case where there is no candidate node w.r.t. the hole. Let $v$ be the certain node w.r.t. the hole.
Because distances of the leaves from the hole are non-increasing,  we can compute the prefix of the caterpillar consisting of
leaves that can be chosen, by binary searching for the largest $i$ such that $d(v,u_{i})\geq \lambda$. Then, we retrieve and return the result
stored for $u_{1},u_{2},\ldots,u_{i}$ (after increasing the number of chosen nodes and, if the certain node is set to NULL, updating it to $v$).

Now consider the case where there is a candidate node w.r.t. the hole, and denote it by $u$. We start with binary searching for $i$ as explained above. As before, $i$ is the largest possible index s.t. $d(v,u_{i})\geq \lambda$, where $v$ is the certain node w.r.t the hole (if it exists). Then,
we check if the distance between $u$ and the certain node nearest to the hole is smaller than $\lambda$ or
$d(u_{i},r)>d(u,r)$, and if so return the result stored for $u_{1},u_{2},\ldots,u_{i}$. Then, again because distances of the leaves
to the hole are non-increasing, we can binary search for the largest $i'\leq i$ such that $d(u_{i'},u)\geq \lambda$
(note that this also takes care of pruning leaves $u_{k}$ that are closer to the hole than $u$).
Finally, we retrieve and return the result stored for $u_{1},u_{2},\ldots,u_{i'}$ (after increasing the number of chosen nodes appropriately and possibly updating the candidate and the certain node).

We process every inactive fragment in $O(\log b)$ time and every active fragment in $O(b)$ time, so the total time for a feasibility test is
$O(\frac{n}{b} \cdot \log b) + O(\frac{n}{b^2} \cdot b) = O(\frac{n}{b}\cdot \log b)$.

\subsection{The algorithm for the optimization problem} \label{sec:nloglogn}

The general idea is to use a heavy path decomposition to solve the optimization  problem
with $O(\log^{2}n)$ feasibility tests. The {\em heavy edge} of a non-leaf node of the tree is the edge leading to the child
with the largest number of descendants. The heavy edges define a decomposition of the nodes into heavy paths. A heavy
path $p$ starts with a head $\head(p)$ and ends with a tail $\tail(p)$ such that $\tail(p)$ is a descendant of $\head(p)$,
and its depth is the number of heavy paths $p'$ s.t. $\head(p')$ is an ancestor of $\head(p)$. The depth 
is always $O(\log n)$~\cite{Sleator1983}.

We process all heavy paths at the same depth together while maintaining, as before, an interval $[\lambda_{1},\lambda_{2})$ such that
$\lambda_{1}$ is feasible and $\lambda_{2}$ is not, which implies that the sought $\lambda^{*}$ belongs to the interval.
The goal of processing the heavy paths at depth $d$ is to further shrink the interval so that, for any heavy path $p$ at depth $d$,
the result of running the feasibility
test on any subtree rooted at $\head(p)$ is the same for any $\lambda\in(\lambda_{1},\lambda_{2})$ and therefore can
be already determined. We start with the heavy paths of maximal depth and terminate
with $\lambda^{*}=\lambda_{1}$ after having determined the result of running the feasibility test on the whole tree.

Let $n_{d}$ denote the total size of all heavy paths at depth $d$. For every such heavy path we construct a caterpillar by replacing
any subtree that hangs off by the certain and the candidate node (this is possible, because we have already determined the result of
running the feasibility test on that subtree). To account for the possibility of including a node of the heavy path in the solution,
we attach an artificial leaf connected with a zero-length edge to every such node.
The caterpillar is then pruned similarly to steps (\ref{removing certain nodes})-(\ref{making distances from the hole monotone})
from Section~\ref{Pre-Processing Fragments}, except that having found the nearest certain node for every candidate
node we cannot simply compare their distance to $\lambda$. Unlike the situation in Section~\ref{Pre-Processing Fragments}, where any pairwise distance in the caterpillar is also a distance in an inactive fragment, and is either smaller or equal to $\lambda_1$ or larger or equal to $\lambda_2$, now such a distance could be within the interval $(\lambda_1,\lambda_2)$. Therefore we create an $1\times 1$ matrix storing the relevant
distance for every candidate node. Then, we apply Theorem~\ref{Frederickson's theorem} with $p=0$ to the obtained set of
$O(n_{d})$ matrices of dimension $1\times 1$. This allows us to determine, using only $O(\log n)$ feasibility tests and
$O(n_{d})$ time exclusive of the feasibility tests, which distances are larger than $\lambda^{*}$, so that we can prune
the caterpillars and work only with the remaining candidate nodes. Then, for every caterpillar we create a row- and column-sorted matrix
storing pairwise distance between its leaves. By applying Theorem~\ref{Frederickson's theorem} with $p=0$ on the obtained set of
square matrices of total side length $O(n_{d})$ we can determine, with $O(\log n)$ feasibility tests and $O(n_{d})$ time
exclusive of the feasibility tests, which distances are larger than $\lambda^{*}$. This allows us to run the bottom-up
procedure described in Section~\ref{linear F.T.} to produce the candidate and the certain node for every subtree rooted
at $\head(p)$, where $p$ is a heavy path at depth $d$.

All in all, for every depth $d$ we spend $O(n_{d})$ time and execute $O(\log n)$ feasibility tests. Summing over all depths,
this is $O(n)$ plus $O(\log^{2}n)$ calls to the feasibility test. Setting $b=\log^{2}n$, we have that the preprocessing time is $O(n \log b) = O(n \log 
\log n)$, and our sublinear feasibility test runs in $O(\frac{n}{b} \cdot \log b) = O(\frac{n}{\log ^2 n} \cdot \log \log n)$, which implies $O(n+\frac{n}{\log^{2}n}\cdot\log\log n\cdot \log^{2}n)=O(n\log\log n)$ time for solving the dispersion problem.

\section{A Linear Time Algorithm for the Dispersion problem}\label{linear algorithm for the dispersion optimization problem}

We first present an $O(n \log^* n)$ time algorithm that uses the procedures of the $O(n \log \log n)$ time algorithm (from Section~\ref{sublinear f.t.}) iteratively, and then move to the more complex linear time algorithm.

\subsection{An \texorpdfstring{\boldmath$O(n\log^{*}n)$}{O(nlog^*n)} time algorithm for the dispersion optimization problem} \label{sectionlog*}
The high level idea of the $O(n \log \log n)$ algorithm was dividing the input tree into fragments, and preprocessing them using the linear feasibility test (or test for short) to get a sublinear feasibility test. In order to get the complexity down to $O(n \log^*n)$ we will use a similar process, iteratively, with growing fragments size, improving in each iteration the performance of our test.
We start with $O(n/c^{3})$ fragments of size at most $c^3$ for some constant $c$.
We pay $O(n)$ time for preprocessing and get an $O(\frac{n}{c^3} \cdot \log c)$ time test.
We use this test to do the same preprocessing again, this time with $O(n/(2^{c})^{3})$ fragments of size at most $(2^c)^3$.
We now construct an $O(\frac{n}{2^{3c}} \cdot c)$ time test.
After $\log^*n$ such iterations we obtain $O(n/\log^{3}n)$ fragments of size at most $\log ^3n$, and have
an $O(\frac{n \log \log n}{\log ^3n})$ time test, which we use to solve the dispersion optimization problem in $O(n)$ time. Since we have $\log^*n$ iterations and linear time work for each
iteration, the total preprocessing is in $O(n \log^*n)$ time.

We now describe one iteration of this process. Assuming we have obtained an $O(\frac{n}{b^3} \cdot \log b)$ time 
test from the previous step, we now find a good partition of the tree with $O(n/B^{3})$ fragments of size at most $B^3$, 
where $B=2^b$. We perform a heavy path decomposition on each fragment, and process the heavy
paths of a certain depth in all the fragments in parallel, starting with heavy paths of maximal depth, similarly to Section~\ref{sec:nloglogn}. Let us consider a single step of this computation. The paths at lower levels have been processed so that each of them is
reduced to a candidate and a certain node, and so each current path is a caterpillar, that we will now reduce to at most two nodes using Frederickson's method. As we have seen in Section \ref{sublinear f.t.}, after some processing, we can create a sorted matrix 
with linear number of rows and columns that contains all pairwise distances in such a caterpillar. As before, this is done by pruning the caterpillars so that the distances of leaves from the root and the hole are monotone, thus obtaining a matrix that is row and column sorted. For any level of the heavy path decomposition the maximum side length of a matrix is $B^{3}$ (i.e., the parameter $n_j$ in Theorem \ref{Frederickson's theorem} is equal to $B^3$) and the total size of all matrices in the level is at most $n$ ($m=n$).
We use Theorem~\ref{Frederickson's theorem} with the parameter $p$ set to $\frac{n}{B^7}$. We thus need to use the
test $O(\max \lbrace \log(\max_{j} \lbrace n_j \rbrace), \log(\frac{m}{p+1}) \rbrace) = O(\max \lbrace \log (B^{3}), \log(\frac{n}{n/B^{7}} \rbrace) = O(\log B)$ times, and additionally spend linear time exclusive of the feasibility tests. This allows us to eliminate all pairwise distances in the caterpillars of the heavy paths at the current level for most of the fragments. We proceed to compute the candidate and certain node of each path as we did before, and continue with the next level
of the heavy path decompositions, until we are left with only the top heavy path of each fragment.
Since $p>0$, during this preprocessing some fragments become active. Once a fragment becomes active we stop
its preprocessing. 

After the preprocessing is done, we have at most $\log B \cdot \frac{n}{B^7}\leq \frac{n}{B^{6}}$ active fragments (because we have not eliminated at most $\frac{n}{B^7}$ pairwise distances for each level), which we will process in $O(B^{3})$
time each in the new test. For all other fragments we have obtained a pruned caterpillar, so that we can, in linear time, precompute the required information for any candidate and certain node with respect to the hole
as described in Subsection \ref{Pre-Processing Fragments}. Notice that this requires the spine of the obtained caterpillar (i.e. the heavy path of depth zero in the decomposition) to also be the spine of the fragment. This can be ensured by simply tweaking the heavy path decomposition so that the first heavy path is always the root-to-hole path in every fragment. This change does not affect the maximal depth of a heavy path, and allows us to preprocess the pruned caterpillars, so that inactive fragments can be processed in logarithmic time in the new test. Thus we have obtained an $O(\frac{n}{B^3} \cdot \log B)$ time test.

We repeat this step iteratively until we have an $O(\frac{n}{\log ^3n} \cdot \log \log n)$ time test, and then use it to
solve the dispersion optimization problem in linear time as explained in Section~\ref{sec:nloglogn}, that is
by performing a heavy path decomposition on the whole tree and applying Theorem~\ref{Frederickson's theorem} with $p=0$.

In each step, accounting for all the levels, we use $O(\log^{2}B)$ calls to the previous test,
which costs $O(\log ^2B \cdot \frac{n}{b^3} \cdot \log b)=O(\frac{n}{b} \cdot \log b)$.
Additionally, we need linear time to construct a good partition, find all the heavy path decompositions,
and apply Theorem~\ref{Frederickson's theorem} on the obtained matrices. Thus, overall
a single iteration takes $O(n)$ time, so summing up over all iterations $O(n\log^{*}n)$.
Actually, the total time spent on calls to the previous tests sums up to only $O(n)$ over all iterations, which will
be important for the improved algorithm presented in the following subsections.

\subsection{Overview of the linear time algorithm}
In the rest of the section we describe how to improve the running time to $O(n)$, giving an optimal solution for the dispersion optimization problem.
Now we cannot afford to divide the tree into fragments separately in every iteration, as this would already add up to superlinear time.
Therefore, we will first modify the construction so that the partition (into small fragments) in the previous iteration is a refinement of the partition
(into large fragments) in the current iteration.

Consider a large fragment that is composed of multiple small fragments. Notice that some of these small fragments are active fragments, i.e., have not been reduced to caterpillars in the previous iteration, but most are inactive fragments, which have
been reduced to caterpillars. The large fragment contains some small fragments whose roots and holes are spine nodes of the large fragment,
and other fragments that form subtrees hanging off of the large fragment's spine (each hanging subtree may
contain several small fragments), see Figure \ref{figure of small fragments inside a large fragment}. This will allow us to
reuse some of the precomputation done in the previous iteration.

\begin{figure}[h]
\begin{center}
\includegraphics[scale=1]{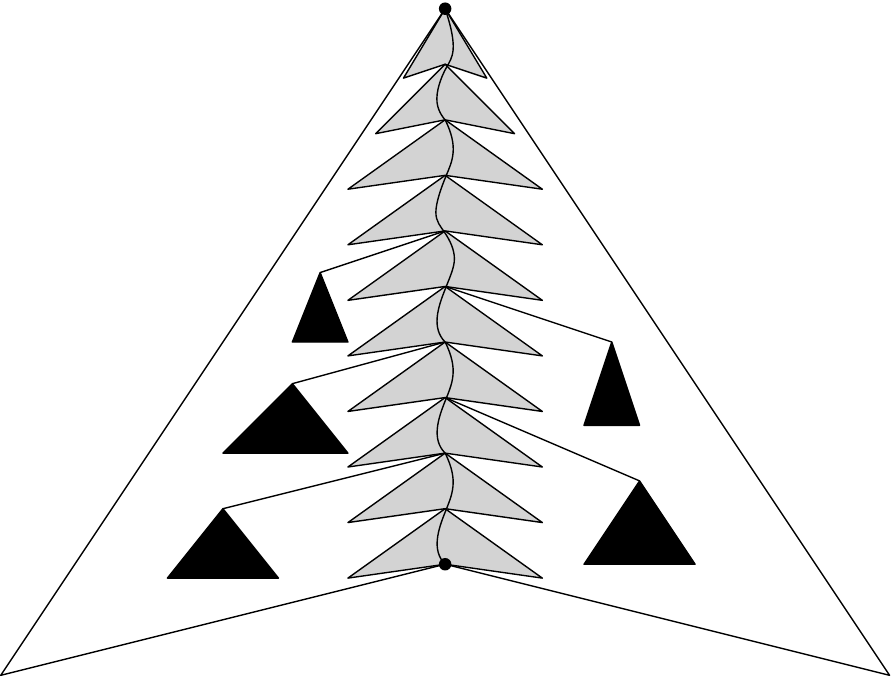}
\end{center}
\caption{A large fragment containing a number of smaller fragments. The small fragments on the spine are gray, and the subtrees hanging off of the spine are black (each of them may contain many small fragments). 
\label{figure of small fragments inside a large fragment}}
\end{figure}

The partitioning used in the algorithm is presented in Subsection \ref{section:partioioning}. In Subsection~\ref{section:lemma1} we show how to reduce the hanging subtrees to at most two nodes, one candidate and one certain. Then, in Subsection~\ref{section:lemma2} we reduce the small active fragments on the spine to caterpillars. This turns the entire fragment into a large caterpillar. In Subsection~\ref{section:lemma3} we prune this large caterpillar so that it becomes monotone and without collisions. Finally, in Subsection~\ref{section:lemma4} we precompute
the resulting large caterpillar for any candidate and certain node with respect to the hole.

\subsection{Partitioning into fragments}\label{section:partioioning}
The following lemma is similar to Lemma 3 in \cite{Frederickson1985}, except that we are working with rooted trees
and our definition of a fragment is more restrictive. It can be also inferred, with some work, from the properties of top trees~\cite{TopTrees}.

\begin{lemma}\label{good partition refinement lemma}
Define a sequence $b_{1},b_{2},\ldots,b_{s}$ of length $s=O(\log^{*}n)$, such that $b_{1}=\Theta(1)$, $b_{s} = \Theta(\log n)$,
and $b_{\ell+1} = 2^{b_{\ell}}$ for $\ell=1,2,\ldots,s-1$. Then, given any binary tree on $n$ nodes,
we can construct in linear time a hierarchy of $s$ good partitions, such that 
the partition at level $\ell$ of the hierarchy has $O(n/{{b_{\ell}}^4})$ fragments, each of size $O({b_{\ell}}^5)$,
and every fragment is contained inside a fragment in the partition at level $\ell+1$.
\end{lemma}
\begin{proof}
We would like to construct the hierarchy of partitions by repeatedly applying Lemma \ref{basic partitioning lemma}. For convenience we slightly tweak  the definition of a good partition, so that for a parameter $b$ there are at most $n/b$ fragments, each of them of size $O(b)$ (as opposed to the original definition, where we had $O(n/b)$ fragments of size at most $b$). We can achieve such a partitioning by applying Lemma \ref{basic partitioning lemma} with the parameter set to $\alpha\cdot b$, where $\alpha$ is the constant hidden in the statement of Lemma~\ref{basic partitioning lemma}. This results in a partition with at most $n/b$ fragments, each of size at most $\alpha \cdot b$. 

We now show how to find a good partition of the tree into at most $n/(2^b)^{4}$ fragments each of size at most $\alpha^{\ell+1}\cdot (2^{b})^{4}$, in $O(n/b^{4})$ time, given a good partition of the tree into at most $n/b^4$ fragments each of size at most $\alpha^{\ell}\cdot b^{4}$, s.t. every fragment in the first partition is contained inside a fragment in the second partition.

Define $T'$ as the tree obtained by collapsing each fragment of the given partition into a single node. Partition $T'$ with the parameter set to $B= {(2^{b})^{4}}/{b^{4}}$. We obtain at most $\frac{n/b^4}{B}=n/(2^{b})^{4}$ fragments. Each of the new large fragments corresponds to a fragment of the original tree of size at most $\alpha\cdot \alpha^{\ell}\cdot b^{4}\cdot B = \alpha^{\ell+1}\cdot (2^{b})^{4}$. Clearly, it holds that each fragment of the given partition is contained in a single large fragment.

The entire hierarchy is obtained by applying this process $\log^{*}n$ times. 
If $b_{\ell}$ denotes the value of the parameter $b$ in the $\ell$-th application, then  $b_{\ell+1}=2^{b_{\ell}}$. The time to construct the $(\ell+1)$-th partition is $O(n/(b_{\ell})^{4})=O(n/b_{\ell})$. The overall time to construct all partitions is therefore 
$$ O(\sum_{\ell} {n}/{b_{\ell}}) = O(\sum_{\ell} {n}/{2^{\ell}}) = O(n).$$
Observe that the size of a (large) fragment in the $(\ell+1)$-th partition is at most $\alpha^{\ell+1} \cdot (b_{\ell+1})^{4}$ which is at most $O((b_{\ell+1})^5)$ for a large enough $\ell$ (i.e., an $\ell$ s.t. $\alpha^{\ell+1}\leq 2^{2^{\ell+1}}$ and $b_{\ell+1}\geq 2^{2^{\ell+1}}$, so
this $\ell$ depends only on the constant $\alpha$).

\end{proof}
 
\noindent Setting $b=b_{\ell}$ and $B=b_{\ell+1}$  we conclude that every large fragment of the $(\ell+1)$-th partition is of size $O(B^5)$ and consists
of small fragments of the $\ell$-th partition (each of size $O(b^5)$).
The goal of the $(\ell+1)$-th iteration will be now to preprocess all $O(n/B^{4})$ large fragments, so that we obtain an $O(\frac{n}{B^{4}}\log B)$
time test. To this end, we will preprocess all inactive fragments, so that they can be skipped in $O(\log B)$ time. For the remaining active
large fragments, we guarantee that there are only $O(n/B^{9})$ of them, so processing each in $O(B^{5})$ time leads
to the claimed query time. To make sure that the total preprocessing time is linear, we will reuse the preprocessing already done
for the $O(n/b^{4})$ small fragments (by induction, only $O(n/b^{9})$ of them are active, while the others are preprocessed).

\subsection{Reducing a hanging subtree into two nodes}\label{section:lemma1}
The first step in reducing a large fragment into a caterpillar is reducing the hanging subtrees of the large fragment into at most two nodes. The following lemma shows that this can be done in total $O(n)$ time for all fragments in all the $\log^* n$ partitions.

\begin{lemma}\label{lemma1}
The hanging subtrees of the large fragments in all the partitions can be reduced to at most two nodes (one candidate and one certain) in total $O(n)$ time. 
\end{lemma}
\begin{proof}
We perform a heavy path decomposition on each hanging subtree. Since each such subtree is contained in a large fragment, its 
size is bounded by $O(B^5)$, and so we have at most $O(\log B)$ levels in the heavy 
path decompositions. On each level we run Frederickson's search (similarly to what we have presented before).
We run this search in parallel, for all heavy paths from all these hanging subtrees in all large fragments.
The number of  nodes in each level for all the heavy paths together is at most $m=n$ and we set
the parameter $p$ (in Theorem~\ref{Frederickson's theorem}) to $\frac{n}{B^{10}}$, so the number of calls to the $O(\frac{n}{b^4} \cdot \log b)$ time feasibility test
is $O(\max \lbrace \log(\max_{j} \lbrace n_j \rbrace), \log(\frac{m}{p+1}) \rbrace) = O(\max \lbrace \log (B^{5}), \log(\frac{n}{n/B^{10}} \rbrace) = O(\log B)$ per level, so
$O(\frac{n}{b^4} \cdot \log b \cdot \log^{2} B) = O(\frac{n}{b})$ for every iteration.
Summing over all $\log^{*}n$ iterations, this is $O(\sum_{\ell}\frac{n}{b_\ell}) = O(n)$.

Constructing the heavy path decompositions and then running Frederickson's search takes linear time in the number
of nodes in the subtrees. If all of them were then reduced to at most two nodes, one candidate and one certain
node for each subtree, this would amortize to $O(n)$ time when summed over all the $\log^{*}n$  iterations. However, we have set
the parameter $p$ to $\frac{n}{B^{10}}$. Thus, considering all the levels, there
might be up to $\log B\cdot\frac{n}{B^{10}}\leq \frac{n}{B^{9}}$ large fragments in which we cannot fully reduce
all the hanging subtrees. We declare them the active large fragments and do not preprocess further.
Considering the total preprocessing time when summed over all the iterations, it is $O(n)$
plus the total size of the active fragments in all iterations, which is $O(\sum_{\ell} \frac{n}{b_{\ell}^{9}}\cdot b_{\ell}^{5})=O(\sum_{\ell} n/b_{\ell})=O(n)$.
\end{proof}

\subsection{Reducing a small active fragment into a caterpillar}\label{section:lemma2}
Having processed the hanging subtrees, we now deal with the active small fragments on the spine of the large fragment. 

\begin{lemma}\label{lemma2}
	The small active fragments on the spines of all large fragments can be reduced to caterpillars in total $O(n)$ time. 
\end{lemma}
\begin{proof}

In the previous iteration we had $O(\frac{n}{b^9})$ active fragments, which is of course
too many for the current level. We take the remaining active small fragments (the ones on spines of the large 
fragments), and reduce them to caterpillars (again by constructing a heavy path decomposition and running
Frederickson's search with $p$ set to $\frac{n}{B^{10}}$). In each fragment we have $O(\log b)$ levels of the heavy
path decomposition, and for each level we need $O(\log B)$ calls to the feasibility test, which takes
$O(\frac{n}{b^4} \cdot \log b \cdot \log B \cdot \log b) = O(\frac{n}{b})$ time in total. We also do linear
work exclusive of the feasibility tests. This processing of small active fragments takes $O(n)$ time over all
$\log^*n$ iterations, by the same reasoning as for the hanging subtrees.
As before, we declare a large fragment active if we were not able to reduce some remaining active small
fragment inside and do not preproces further.
\end{proof}

Note that even though we have reduced the small active fragments to caterpillars, we have not preprocessed them yet
(as opposed to the small inactive fragments, that have been reduced to caterpillars and then further preprocessed).
However, there are only $O(n/b^{9})$ such fragments, each of size $O(b^{5})$, so we can split each
of these caterpillars into caterpillars containing exactly one spine node, that are trivial to preprocess.
The total number of small caterpillars is then still $O(n/b^{4}+n/b^{9}\cdot b^{5})=O(n/b^{4})$,
so it does not increase asymptotically.

\subsection{Pruning the resulting large caterpillar}\label{section:lemma3}
So far, each inactive large fragment was reduced to one large caterpillar consisting of the concatenation of caterpillars from the smaller fragments on its spine. We want
to prune this large caterpillar (as we have done with the inactive fragments before) in order to ensure
that it is monotone and contains no collisions.   Then, we will
be able to preprocess the caterpillar for every possible candidate and certain node with respect to the hole.

\begin{lemma}\label{lemma3}
	The inactive large fragments can be processed in overall $O(n)$ time so that each of them is a caterpillar, the distances of the leaves from the root and the hole are monotone, and there are no collisions between candidate nodes and certain nodes.
\end{lemma}
\begin{proof}

We start by pruning the large caterpillar so that distances of the candidates from the root are
monotone. Since the distances are monotone inside each small caterpillar, we only have to check the last (bottom most) 
candidate node of a small caterpillar against a range of consecutive candidates nodes below it in the large
caterpillar, starting from the top candidate node in the next small caterpillar. 
At this point we need to be more precise about how the fragments are represented. Each inactive small
fragment has been reduced to a caterpillar. For each such caterpillar, we maintain a doubly-linked list containing
all of its candidate nodes in the natural order, and additionally we store the first and the last certain
node (nearest to the root and to the hole, respectively). Assuming such representation, we traverse
the large caterpillar from bottom to top while
maintaining a stack with the already processed small caterpillars that still contain some candidate nodes.
We retrieve the last candidate of the current small caterpillar and eliminate the appropriate candidate
nodes below it. This is done by iteratively retrieving the top small caterpillar from the stack and then
accessing its top remaining candidate node $u$. If the distance of $u$ from the root is too small, we
remove $u$ from the front of its doubly-linked list, pop the list from the stack if it becomes empty,
and repeat. Now the distances of the remaining candidates from the root are
non-decreasing. We repeat the symmetric process from top to bottom to ensure that also the distances
from the hole are non-decreasing. The whole procedure takes linear time in the number of eliminated
candidate nodes plus the number of small fragments, so $O(n)$ overall.

Now we need to eliminate all internal ``collisions'' inside the large caterpillar (i.e., remove candidate nodes that
are too close to certain nodes). These collisions can occur between a candidate node of one small caterpillar
and a certain node of another small caterpillar. Each certain node eliminates a range of consecutive
candidates nodes below and above its small caterpillar. Therefore, we only need to consider the 
last certain node of a small caterpillar and a range of consecutive candidate nodes below it in the
large caterpillar, starting from the top candidate node in the next caterpillar (and then repeat
the reasoning for the first certain node and a range of candidates nodes above it). This seems very similar
to the previous step, but in fact is more complex: a certain node $v$ eliminates a candidate node $u$ when
$d(u,v)<\lambda^{*}$, but we do not know the exact value $\lambda^{*}$ yet! Therefore, we need to run
a feasibility test to determine if $u$ should be eliminated. We cannot afford to run many such tests
independently and hence will again resort to Frederickson's search. Before we proceed, observe that
it is enough to find, for every small caterpillar, the nearest certain node $v$ above it, and then determine
the prefix of the small caterpillar that is eliminated by $v$. To find $v$ for every small caterpillar,
we first traverse the large caterpillar from top to bottom while maintaining the nearest certain node
in the already seen prefix.

To apply Frederickson's search, we would like to create a matrix of dimension $1\times O(b^{5})$
for each of the $O(\frac{n}{b^{4}})$ small caterpillars in the whole tree. The total size of all small
caterpillars is at most $m=n$, so setting $p$ to $\frac{n}{B^9}$ would imply
$O(\max \lbrace \log(\max_{j} \lbrace n_j \rbrace), \log(\frac{m}{p+1}) \rbrace) = O(\max \lbrace \log (b^{5}), \log(\frac{n}{n/B^{9}}) \rbrace = O(\log B)$ calls to the $O(\frac{n}{b^{4}} \cdot \log b)$ time test.
There are two difficulties, though. First, we cannot guarantee that the processing exclusive of the feasibility tests
would sum up to $O(n)$ over all iterations. Second, it is not clear how to provide constant time access
to these matrices, as the candidate nodes inside every small caterpillar are stored in a doubly-linked list,
so we are not able to retrieve any of them in constant time.

We mitigate both difficulties with the same trick. We proceed in $O(\log b)$ steps that, intuitively, correspond
to exponentially searching for how long should the eliminated prefix be. In the $i$-th step we
create a matrix of dimension $1\times 2^{i}$ for each of the still remaining small caterpillars.
The matrix is stored explicitly, that is we extract the top $2^{i}$ candidates nodes from the doubly-linked
list and store them in an array. We set $p$ to $\frac{n}{B^{10}}$ and run Frederickson's search
on these matrices. Then, if all $2^{i}$ candidate nodes have been eliminated, we proceed to the next
step. Otherwise we have already determined the small caterpillar's prefix that should be eliminated. The total time for all feasibility tests is then
$O(\frac{n}{b^{4}}\cdot \log b \cdot \log b \cdot \log B) = O(\frac{n}{b})$. 
Observe that the total length of all arrays constructed during this procedure is bounded by the number
of eliminated candidate nodes multiplied by 4. Consequently, both the time necessary
to extract the top candidate nodes (and arrange them in an array) and the time exclusive of the feasibility tests sums up to $O(n)$
over all iterations. During this process we might declare another $O(\frac{n}{B^9})$ large fragments active
by the choice of $p$.
\end{proof}

\subsection{Preprocessing the large caterpillar for every possible eliminated suffix}\label{section:lemma4}

Having reduced a large fragment to a large caterpillar, we are now left only with preprocessing this large caterpillar for every possible eliminated suffix.
Because we have already pruned the large caterpillar and removed any collisions between a certain
node and a candidate node, we only have to consider the remaining candidate nodes. For convenience we will refer to the remaining candidate nodes attached to the spine simply as nodes, and reserve the term candidate for the nodes stored as candidate nodes of solutions in the precomputed information.

We need to store the caterpillars in a data structure that will allow us to merge caterpillars efficiently, as well as search for the bottom-most node in a caterpillar that is at distance at least $d$ from the hole, for any given $d \le \lambda^*$. 
We therefore store the nodes of each caterpillar in a balanced search tree keyed by the nodes' distance to the hole of the caterpillar. Note that even though each search tree is keyed by distances to a different hole, we can merge two trees in logarithmic time, since they correspond to two small caterpillars where one is above the other (and hence all its keys can be thought of as larger than all the keys of the bottom one). Since we merge search trees of small caterpillars into one  search tree of a large caterpillar, each merge operation costs at most $O(\log B)$, and we have $O(n/b^9)$ such operations (one per small caterpillar), this sums up to $O(\frac{n}{b})$ and $O(n)$ overall.
Maintaining predecessor and successor pointers in the search trees enables the linked list interface that was used in the previous steps, and since in the previous steps we only prune prefixes and suffixes of small caterpillars, we can use split operations for that.

The main goal of the current step is to update the information stored in each node (i.e. the certain node, candidate node, and number of chosen nodes) for any node that might be queried in the future, so that the information holds the solution for the large caterpillar. When doing this, we need to take into account that the information that is currently stored in nodes might be false, even with regard to the appropriate small caterpillar, since we have eliminated some nodes in the previous step.
We next define a few terms and show a lemma that will be used in the rest of this section.

\vspace{0.04in} \noindent {\bf The unstable and interesting affixes.} 
Consider a caterpillar with root $r$ and hole $h$, and let $u_{i}$ denote its $i$-th node. In Lemma~\ref{lemma3} we have already made sure that $d(r,u_{i}) < d(r,u_{i+1})$
and that $d(h,u_{i}) > d(h,u_{i+1})$. We define the \emph{unstable prefix} to
consist of all nodes $u_{i}$ such that $d(r,u_{i}) < \frac{\lambda^{*}}{2}$, and similarly the
\emph{unstable suffix} consists of all nodes $u_{i}$ such that $d(h,u_{i}) < \frac{\lambda^{*}}{2}$.
We have the following property.

\begin{lemma}
\label{stable infix}
If a node of a small caterpillar is removed due to pruning or collision elimination (i.e. due to Lemma~\ref{lemma3}) then it
belongs to the unstable prefix or to the unstable suffix.
\end{lemma}

\begin{proof}
Because of symmetry, it is enough to consider a small caterpillar and all the subsequent
small caterpillars. If a node $u_{i}$ of the former is pruned then there is a top node
$v$ of one of the subsequent caterpillars such that $d(h,u_{i}) < d(h,v)$, where $h$ is the hole
of the large caterpillar. But the distance of $v$ from the spine is less than $\frac{\lambda^{*}}{2}$ (since it is a candidate node),
 and it follows that the distance of $u_{i}$ to the hole of its small caterpillar is also less
than $\frac{\lambda^{*}}{2}$. Hence, $u_{i}$ belongs to the unstable suffix.
If on the other hand $u_{i}$ collides with a certain node $v$ of one of the subsequent
caterpillars then $d(u_{i},v)<\lambda^{*}$. But the distance of $v$ from the spine is at least
$\frac{\lambda^{*}}{2}$, so the distance of $u_{i}$ to the hole of its small caterpillar must be
less than $\frac{\lambda^{*}}{2}$. Again, this means that $u_{i}$ belongs to the unstable suffix.
\end{proof}

We define the \emph{interesting prefix} of a caterpillar to be all the nodes whose distance from the top-most node is less than $\lambda^*$.
Similarly, the \emph{interesting suffix} contains all the nodes whose distance from the bottom-most node is less than $\lambda^*$. The interesting suffix contains one additional node -- the lowest node whose distance from the bottom-most node is greater or equal to $\lambda^*$. If such a node does not exist, then the interesting suffix as well as the interesting prefix are the entire caterpillar. We call such a caterpillar \emph{fresh}. Note that the unstable prefix (suffix) of any caterpillar is contained in its interesting prefix (suffix).

\paragraph{The preprocessing procedure.}
We now describe the preprocessing procedure on a given large caterpillar. We start by checking if the large caterpillar is fresh, i.e. if the distance between the bottom-most node and the top-most node is smaller than $\lambda^*$. This can be
done by applying Frederickson's search in parallel for all large caterpillars with $p$ set to
be $\frac{n}{B^{9}}$, using $O(\frac{n}{b^{4}}\cdot \log b \cdot \log B)=O(\frac{n}{b})$
total time for the feasibility tests and creating up to $O(\frac{n}{B^{9}})$ additional active large
fragments. If we discover that a large caterpillar is fresh, there is nothing to update. 
Otherwise, we gather all nodes in the large caterpillar belonging to
an interesting prefix or suffix of some small caterpillar and call them \emph{relevant}. In particular,
all nodes of a fresh small caterpillar are relevant.
We arrange all the relevant nodes in an array and apply Frederickson's search to determine whether the distance between any pair of them is smaller than $\lambda^*$. We again
set $p$ to be $\frac{n}{B^{9}}$, need $O(\frac{n}{b^{4}}\cdot \log b \cdot \log B)=O(\frac{n}{b})$
time for the feasibility tests, and create $O(\frac{n}{B^{9}})$ additional active large fragments.
The time exclusive of the feasibility tests is bounded by the number of relevant nodes.

We next iterate over the small caterpillars from top to bottom, and process each of them while maintaining the following invariant:
\begin{invariant} \label{invariant small caterpillar is correct after prep.} After we process a small caterpillar, each node in its interesting suffix stores the correct information for the case where this node is the bottom-most chosen node in the large caterpillar. 
\end{invariant}

We now describe how to process each small caterpillar assuming the invariant holds for all small caterpillars above it. We then show that the invariant also holds for the just processed small caterpillar. 

\begin{noindlist}
\item \textbf{preprocessing the interesting prefix:}
Before processing a small caterpillar, each node in its interesting prefix stores only itself as the solution for the case where this node is the bottom-most chosen node (since inside the current small caterpillar only one node can be taken from the interesting prefix). In order for us to store the correct information with respect to the large caterpillar (i.e. after concatenating the current small caterpillar with the ones above it) we need to find the bottom-most node in the small caterpillars above s.t. its distance to the current node is greater or equal to $\lambda^*$.
The node we seek is in the interesting suffix of the small caterpillar above the current one. If the above caterpillar is fresh, we also need to look at the interesting suffix of the small caterpillar above it (which might also be fresh) and so on. 
We iterate over the interesting prefix of the current small caterpillar top to bottom. 
For each node in this interesting prefix, we find the closest node at distance at least $\lambda^*$ above it (if it exists), and update the the information stored with the current node of the interesting prefix, using the information stored in the node we found (i.e. add the number of chosen nodes, and copy the certain node and candidate node stored in it). Note that we only look at relevant nodes here, and for each pair of relevant nodes we already know whether the distance between them is smaller than $\lambda^*$ or not.

Since Invariant \ref{invariant small caterpillar is correct after prep.} holds for every small caterpillar above the current one, we now have the correct information stored for every node of the interesting prefix.

\item \textbf{preprocessing the interesting suffix:}
Consider a node of the interesting suffix. Denote the node stored as the candidate node of the solution where this is the bottom-most chosen node by $x$. If there is no candidate node, let $x$ denote the certain node of this solution.
We check if the information stored in node $x$ was updated when we processed the interesting prefix. If it was updated, we update the information stored in the current node (of the interesting suffix) accordingly, and move on to the next node in the interesting suffix. 

If the information stored in node $x$ was not updated, then $x$ is not in the interesting prefix of the current small caterpillar, which implies (by definition of the interesting prefix) that it is the certain node stored for this solution. Moreover, since $x$ does not belong to the interesting prefix, there was a node $y$ above $x$ in the current small caterpillar which was at distance greater or equal to $\lambda^*$ from $x$. Therefore, $y$ had been stored as the candidate node of the solution, and was pruned due to Lemma \ref{lemma3}. Hence, we need to find the bottom-most remaining node above $x$ which is at distance greater or equal to $\lambda^*$ from it. We claim that this node is simply the bottom-most remaining node above the current small caterpillar. Denote this node by $z$. We show that $d(x,z) \ge d(x,y)$ (implying that $d(x,z) \ge \lambda^*$).
We consider the three possible cases:
\begin{itemize}
\item $y$ was pruned because it was closer to the root than some node above it, which we call the removing node. The removing node cannot be above $z$, since otherwise it would either imply that $z$ should also be pruned, or that the removing node is not the cause for $y$'s pruning.
Thus, the removing node is either below $z$, or it is $z$ itself. As shown in Step \ref{making distances from the root monotone} of the preprocessing described in Subsection \ref{Pre-Processing Fragments}, the removing node must be farther from any node below $y$, than $y$. Therefore $z$ is farther from $x$ than $y$.

\item If $y$ was pruned because it was closer to the hole than some node below it, then since $z$ was not pruned, it must be farther from the hole than $y$, and so it is also farther from $x$ than $y$.
\item If $y$ was pruned due to a collision with some certain node, then in the process of pruning nodes s.t. they are sorted by their distance to the hole (which is done before collision elimination), $y$ and $z$ were not pruned. This implies that $z$ is farther from the hole than $y$, and so it is also farther from $x$.

\end{itemize}

Note that since we always remove an entire prefix or suffix during the pruning process, the first remaining node above $y$ must be in a different small caterpillar. We update the information in the current node according to the already computed information stored in that node, and proceed to the next node in the interesting suffix.

\end{noindlist}

\begin{lemma}
Invariant \ref{invariant small caterpillar is correct after prep.} is maintained during the above preprocessing.
\end{lemma}
\begin{proof}
We assume that the invariant holds for every small caterpillar above the current one, and show that after the preprocessing it also holds for the current one. During the preprocessing we update the information stored in the nodes of the interesting suffix by using the information stored in nodes of previously processed caterpillars (where the information is already correct), so we only need to show that we do not use nodes of the current small caterpillar that were pruned due to Lemma \ref{lemma3}. Any such pruned node is either in the unstable suffix or in the unstable prefix of the small caterpillar. If the node is in the pruned suffix, it cannot be stored in a solution where the bottom-most node is above the pruned suffix, and so any node where it might be stored was also pruned. If the pruned node is in the unstable prefix, its distance to the root is less than $\frac{\lambda^*}{2}$, and so if it is stored in some solution, it must be the candidate node for that solution. Because we always check if the stored candidate node for some node in the interesting suffix still exists, we do not use such pruned nodes in the solutions we store.
\end{proof}

We now have the correct information stored in every node of the interesting suffix of the large caterpillar. These are the only nodes that might be queried when running a feasibility test (since a node above the interesting suffix cannot be the bottom-most chosen node in the caterpillar).

\paragraph{The preprocessing analysis.}
It remains to analyze the running time of this preprocessing procedure. We claim that it is amortizes $O(n)$. 
In every step, we spend time which is linear in the number of relevant nodes. Some of these nodes will not be relevant in the next iteration, but some of them might be relevant for many iterations. Observe that in each iteration, unless the large caterpillar is fresh (in which case we do not iterate over the relevant nodes), the nodes that will be relevant in the next iteration are only the ones in the interesting suffix and prefix of the resulting large caterpillar. The interesting prefix (and similarly the interesting suffix) of the resulting large caterpillar
consists of nodes belonging to fresh small caterpillars and at most one interesting prefix of
a small caterpillar. The nodes belonging to fresh small caterpillars have not been processed as relevant nodes until now, but the nodes belonging to interesting prefixes (and suffixes) of small caterpillars have been, and such nodes might be relevant for many future iterations. Fortunately, since the number of nodes in an interesting prefix or suffix is bounded by the number of nodes in a small caterpillar, which is $O(b^{5})$, and since there is at most one interesting prefix (suffix) of a small caterpillar whose nodes will stay relevant, this sums up to $O(\frac{n}{B^{4}}\cdot b^{5}) = O(\frac{n}{b})$ per iteration, and $O(n)$ overall.

To conclude, we have shown how to obtain an $O(\frac{n}{B^{4}}\log B)$ time test given an $O(\frac{n}{b^{4}}\log b)$
time test, in $O(\frac{n}{b})$ time plus amortized constant time per
node. After $\log^*n$ iterations, we have an  $O(\frac{n}{\log ^4n} \cdot \log \log n)$ test,
which we use to solve the dispersion optimization problem in linear time. All iterations of the preprocessing
take $O(n)$ time altogether.

\section{The Weighted Dispersion Problem}\label{section:weighted}

In this section we present an $O(n\log n)$ time algorithm for the weighted search problem and prove that it is optimal, by giving an 
$\Omega(n \log n)$ lower bound in the algebraic computation tree model.

Recall that $f(P)=\min_{u,v\in P}\{d(u,v)\}$. The weighted dispersion problems are defined as follows.
\begin{itemize}
\item  {\em The Weighted Dispersion Optimization Problem.}  Given a tree with non-negative edge lengths, non-negative node weights, and a number $W$, find a subset $P$ of nodes of weight at least $W$ such that  $f(P)$ is maximized. 
\item  {\em The Weighted Dispersion Search Problem (feasibility test).}  Given a tree with non-negative edge lengths, non-negative node weights, a number $W$, and a number $\lambda$, find a subset $P$ of nodes of weight at least $W$ such that  $f(P)\geq\lambda$, or declare that no such subset exists. 
\end{itemize}
As explained in the introduction, an $O(n \log n)$ time algorithm for the feasibility test implies an $O(n\log^{2}n)$ time solution for the optimization problem.

\subsection{An \texorpdfstring{\boldmath$ \Omega (n\log n)$}{Omega(nlogn)} lower bound for the weighted feasibility test}\label{weighted f.t. lower bound}

In the {\em Set Disjointness} problem, we are given two sets of real numbers $X=\lbrace x_1,x_2,\ldots,x_n \rbrace$ and $Y=\lbrace y_1,y_2,\ldots,y_n \rbrace$ and we need to return true iff $X \cap Y = \emptyset$. 
It is known that solving the set disjointness problem requires $\Omega(n \log n)$ operations in the algebraic
computation tree model~\cite{BenOr}.
We present a reduction from the set disjointness problem to the weighted dispersion search problem. We assume without loss of generality that the elements of $X$ and $Y$ are non-negative (otherwise,
we find the minimum element and subtract it from all elements).

Given the two sets $X$ and $Y$, we construct a tree $T$ as follows. We first compute and set
$K := 2 \cdot \max (X\cup Y) +3$.
The tree $T$ contains two vertices of weight zero, $u$ and $v$, that are connected by an edge of length $\frac{K}{2}$. In addition, for each element $x_i \in X$, we have a node of weight $x_i+1$, that is connected to $u$ by an edge of length $\frac{K}{2} - x_i -1$. For each element $y_i \in Y$, we have a node of weight $K - y_i -1$, that is connected to $v$ by an edge of length $y_i+1$.
\begin{lemma}\label{lemma of the reduction to set disjointness}
$X \cap Y = \emptyset$ iff the weighted feasibility test returns false for $T$ with  parameters $W=K$ and $\lambda =K$.
\end{lemma}

\begin{proof}
We start by proving that if $X \cap Y \neq \emptyset$ then there is a subset of the vertices of the tree, $P \subseteq T$, s.t. $W(P) \geq K$ and $f(P) \geq K$ (where $W(P)$ is the sum of the weights of all the vertices in $P$). Suppose that  $x_i=y_j$. Then, $d(x_i,y_j) = d(x_i,u) + d(u,v) + d(v,y_j) = \frac{K}{2} - x_i -1 + \frac{K}{2} + y_j +1 = K \geq K$. Furthermore, $W(x_i) + W(y_j) = x_i + 1 + K- y_j -1 = K \geq K$, and so the feasibility test should return true due to $P = \lbrace x_i, y_j \rbrace$.

Now we need to prove that if the feasibility test returns true, the sets are not disjoint.
We start by proving that any feasible subset $P$ cannot have more than one vertex from each of the sets.
Assume for contradiction that $P'$ is a feasible subset that contains at least two vertices that correspond to elements of $X$. Denote these two vertices by $x_i$ and $x_j$. We assume that $P'$ is a feasible solution, and so $d(x_i,x_j) \geq K$. But, $d(x_i,x_j) = \frac{K}{2} - x_i -1 + \frac{K}{2} - x_j -1 = K - x_i - x_j -2 < K$ (since $x_i$ and $x_j$ are non-negative) and we have a contradiction.
If $P'$ contains two elements of $Y$, denoted $y_{i}$ and $y_{j}$, we also obtain a contradiction because
$d(y_{i},y_{j})=y_{i}+1+y_{j}+1<K$.
Observe that it is also not possible for $P$ to contain only one node, since the weight of any node is strictly less than $K$. So, we now only need to prove that if $P= \{ x_i,y_j \}$ is a feasible solution, then $x_i = y_j$.
If $P= \{ x_i,y_j \}$ is a feasible solution then on the one hand $d(x_i,y_j) = \frac{K}{2} - x_i -1 + \frac{K}{2} + y_j +1 = K - x_i + y_j \geq K$, implying that $y_j \geq x_i$. On the other hand, $W(P) = x_i + 1 + K - y_j -1 = K + x_i - y_j \geq K$, implying that $x_i \geq y_j$. We conclude that $x_i = y_j$.
\end{proof}

\subsection{An \texorpdfstring{\boldmath$ O(n\log n)$}{O(nlogn)} time algorithm for the weighted feasibility test}\label{weighted f.t.}

Similarly to the unweighted case, we compute for each node of the tree, the subset of nodes $P$ in its subtree  s.t. $f(P) \geq \lambda$ and the total  weight of  $P$ is maximized. We compute this by going over the nodes of the tree bottom-up. Previously, the situation was simpler, as for any subtree we had just one candidate node (i.e., a node that may or may not be in the optimal solution for the entire input tree). This was true because nodes had uniform weights. Now however, there could be many candidates in a subtree, as the certain nodes are only the ones that are at distance at least $\lambda$ from the root of the subtree (and not $\frac{\lambda}{2}$ as in the unweighted case).

Let $P$ be a subset of the nodes in the subtree rooted at $v$, and  $h$ be the node in $P$ minimizing $d(h,v)$. We call $h$ the {\em closest chosen node} in $v$'s subtree. In our feasibility test, $v$ stores an optimal solution $P$ for each possible value of $d(h,v)$ (up to $\lambda$, otherwise the closest chosen node does not affect nodes outside the subtree). That is, a subset of nodes $P$ in $v$'s subtree, of maximal weight, s.t. the closest chosen node is at distance {\em at least} $d(h,v)$ from $v$, $f(P) \ge \lambda$. 
This can be viewed as a monotone polyline, since the weight of $P$ (denoted $W(P)$) only decreases as the distance of the closest chosen node increases (from 0 to $\lambda$). $W(P)$ changes only at certain points called {\em breakpoints} of the polyline. Each point of the polyline is a key-value pair, where the key is $d(h,v)$ and the value is $W(P)$. We store with each breakpoint the value of the polyline between it and the next breakpoint,
i.e., for a pair of consecutive breakpoints with keys $a$ and $a+b$, the polyline value of the interval $(a,a+b]$ is associated with the former.
The representation of a polyline consists of its breakpoints, and the value of the polyline at key 0.

\begin{figure}[h]
\begin{center}
\includegraphics[scale=0.55]{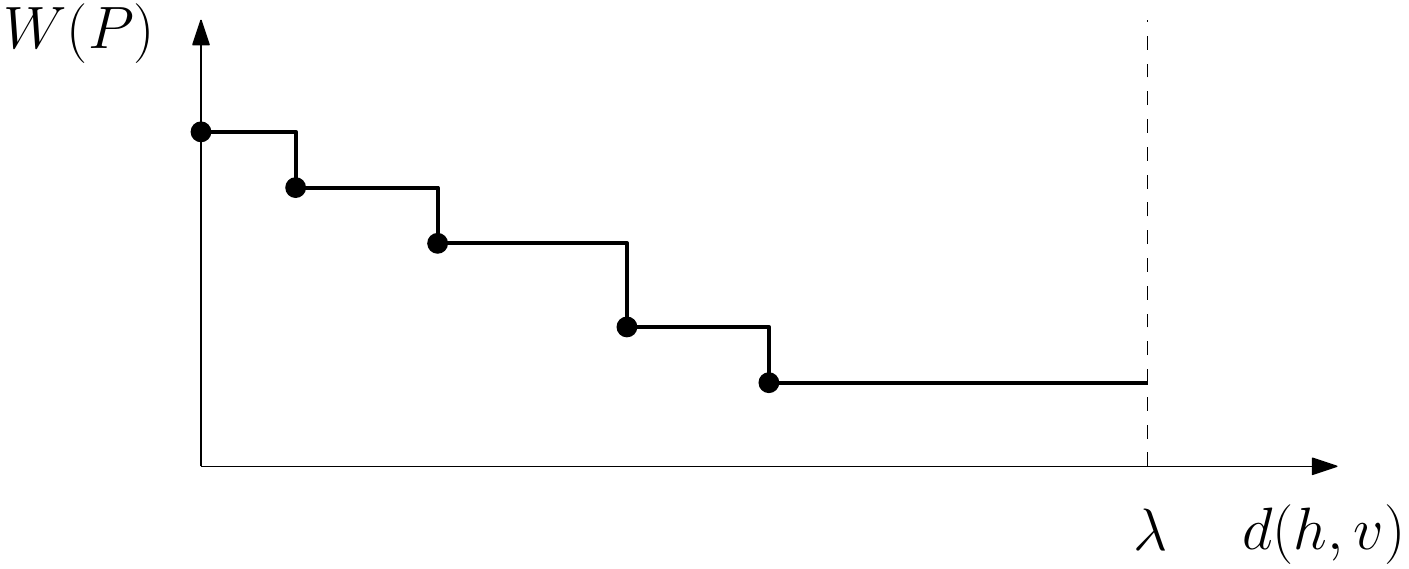}
\end{center}
\caption{The polyline represents the weight of the optimal solution $P$ as a function of the distance of the closest chosen node in the subtree. 
The weight of $P$ only decreases at certain points called breakpoints (in bold). Each breakpoint stores the value in the interval between itself and the next breakpoint.
\label{figure of a polyline}}
\end{figure}

The algorithm computes such a polyline for the subtrees rooted at every node $v$ of the tree by merging the polylines computed for the subtrees rooted at $v$'s children. We assume w.l.o.g. that the input tree is binary (for the same reasoning as in the unweighted case), and show how to
implement this step in time $O(x \log (\frac{2y}{x}))$, where $x$ is the number of breakpoints in the polyline with fewer breakpoints, and $y$ is the number of breakpoints in the other polyline.

\subsubsection{Constructing a polyline.} We now present a single step of the algorithm. We postpone the discussion of the data structure used to store the polylines for now, and first describe how to obtain the polyline of $v$ from the polylines of its children. Then, we state the exact interface of the data structure that allows executing such a procedure efficiently, show how to implement such an interface, and finally analyze the complexity of the resulting algorithm.

If $v$ has only one child, $u$, we build $v$'s polyline by querying $u$'s polyline for the case that $v$ is in the solution (i.e., query $u$'s polyline with distance of the closest chosen node being $\lambda-d(v,u)$), and add to this value the weight of $v$ itself. We then construct the polyline by taking the obtained value for $d(h,v)=0$ and merging it with the polyline computed for $u$, shifted to the right by $d(v,u)$ (since we now measure the weight of the solution as a function of the distance of the closest chosen node to $v$, not to $u$). The value between zero and $d(v,u)$ will be the same as the value of the first interval in the polyline constructed for $u$, so the shift is actually done by increasing the keys of all but the first breakpoint by $d(v,u)$. Note that it is possible for the optimal solution in $v$'s subtree not to include $v$. Therefore we need to check, whether the value stored at the first breakpoint (which is the weight of the optimal solution where $v$ is not included) is greater than the value we computed for the case $v$ is chosen. If so, we store the value of the first breakpoint also as the value for key zero.

If $v$ has a left child $u_1$ and a right child  $u_2$, we have two polylines $p_1$ and $p_2$ (that represent the solutions inside the subtrees rooted at $u_1$ and $u_2$), and we want to create the polyline $p$ for the subtree rooted at $v$. Denote the number of breakpoints in $p_1$ by $x$ and the number of breakpoints in $p_2$ by $y$. Assume w.l.o.g. that $x \leq y$.
We begin with computing the value of $p$ for key zero (i.e. $v$ is in the solution). In this case we query $p_1$ and $p_2$ for their values with keys $\lambda - d(v,u_1)$ and $\lambda - d(v,u_2)$ respectively (if one of these is negative, we take zero instead), and add them together with the weight of $v$. As in the case where $v$ has only one child, it is possible for the optimal solution in $v$'s subtree not to include $v$ itself. Therefore we need to check, after constructing the rest of the polyline, whether the value stored at the first breakpoint is greater than the value we computed for the case $v$ is chosen, and if so, store the value of the first breakpoint as the value for key zero.

It remains to construct the rest of the polyline $p$. Notice that we need to maintain that $d(h_1,h_2) \geq \lambda$ (where $h_1$ is the closest chosen node in $u_1$'s subtree and $h_2$ is the closest chosen node in $u_2$'s subtree). We start by shifting $p_1$ and $p_2$ to the right by $d(v,u_1)$ and $d(v,u_2)$ respectively, because now we measure the distance of $h$ from $v$, not from $u_1$ or $u_2$. We then proceed in two steps, each computing half of the polyline $p$.

\vspace{0.04in} \noindent {\bf Constructing the second half of the polyline.}\label{subsection constructing the second half of the polyline} We start by constructing the second half of the polyline, where $d(h,v) \geq \frac{\lambda}{2}$. In this case we query both polylines with the same key, since $d(h_1,v) \geq \frac{\lambda}{2}$ and $d(h_2,v) \geq \frac{\lambda}{2}$ implies that $d(h_1,h_2) \geq \lambda$. The naive way to proceed would be to iterate over the second half of both polylines in parallel, and at every point sum the values of the two polylines. This would not be efficient enough, and so we only iterate over the breakpoints in the second half of $p_1$ (the smaller polyline). These breakpoints induce intervals of $p_2$. For each of these intervals we increase the value of $p_{2}$ by the value in the interval in $p_1$ (which is constant). See Figure~\ref{figure of constructing the second half of the polyline}. This might require inserting some of the breakpoints from $p_1$, where there is no such breakpoint already in $p_2$. Thus, we obtain the second half of the monotone polyline $p$ by modifying the second half of the monotone polyline $p_{2}$.

\begin{figure}[h]
\begin{center}
\includegraphics[scale=.55]{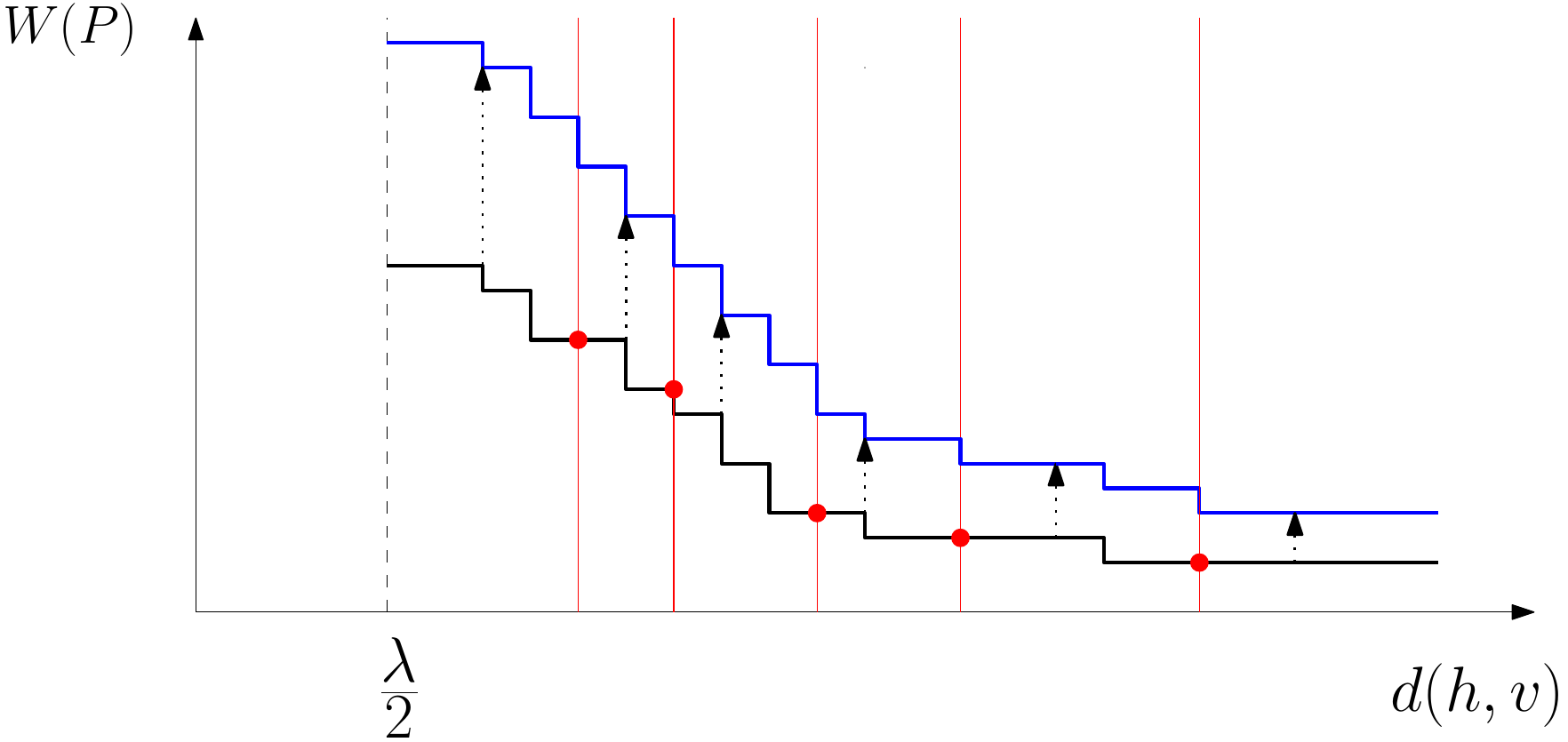}
\end{center}
\caption{Constructing the second half of the polyline $p$ (in blue). The black polyline is $p_2$. The breakpoints inserted from $p_1$ are in red. In each interval (between consecutive red lines) we raise the polyline by the value in $p_1$.
\label{figure of constructing the second half of the polyline}}
\end{figure}

\vspace{0.04in} \noindent {\bf Constructing the first half of the polyline.} We now need to consider two possible cases: either $d(h_1,v) < d(h_2,v)$ (i.e. the closest chosen node in $v$'s subtree is inside $u_1$'s subtree), or $d(h_1,v) > d(h_2,v)$ ($h$ is in $u_2$'s subtree). Note that in this half of the polyline $d(h,v)<\frac{\lambda}{2}$, and therefore $d(h_1,v) \neq d(h_2,v)$. For each of the two cases we will construct the first half of the polyline, and then we will take the maximum of the two resulting polylines at every point, in order to have the optimal solution for each key.

\vspace{0.04in} \noindent {\bf Case I: \boldmath$d(h_1,v) < d(h_2,v)$.} Since we are only interested in the first half of the polyline, we know that $d(h_1,v) < \frac{\lambda}{2}$. Since $d(h_2,v) +d(h_1,v)\geq \lambda$ we have that  $d(h_2,v) > \frac{\lambda}{2}$. Again, we cannot afford to iterate over the breakpoints of $p_2$, so we need to be more subtle.

We start by splitting $p_1$ at $\frac{\lambda}{2}$ and taking the first half (denoted by $p_1'$). We then split $p_2$ at $\frac{\lambda}{2}$
and take the second half (denoted by $p_2'$). Consider two consecutive breakpoints of $p_1'$ with keys $x$ and $x+y$. We would like to increase the value of $p_1'$ in the interval $(x,x+y]$ s.t. the new value is the maximal weight of a valid subset of nodes from \emph{both} subtrees rooted at $u_1$ and $u_2$, s.t. $x < d(h_1,v) \leq x+y$. Therefore $d(h_2,v) \ge \lambda-x-y$. $p_2'$ is monotonically decreasing, and so we query it at $\lambda-x-y$, and increase by the resulting value.

This process might result in a polyline which is not monotonically decreasing, because as we go over the intervals of $p_1'$ from left to right we increase the values there more and more.
To complete the construction, we make the polyline monotonically decreasing by scanning it from $\frac{\lambda}{2}$ to zero and deleting unnecessary breakpoints. We can afford to do this, since the number of breakpoints in this polyline is no larger than the number of breakpoints in $p_1$.
Note that we have assumed we have access to the original data structure representing $p_{2}$, but this structure has been  modified to obtain the second half of $p$. However, we started with computing the second half of $p$ only to make the description simpler. We can simply start with the first half.

\vspace{0.04in} \noindent {\bf Case II: \boldmath$d(h_1,v) > d(h_2,v)$.}
Symmetrically to the previous case, we increase the values in the intervals of $p_2$ induced by the breakpoints of $p_1$ by the appropriate values of $p_{1}$ (similarly to what we do in the construction of the second half of the polyline). Again, the resulting polyline may be non-monotone, but this time we cannot solve the problem by scanning the new polyline and deleting breakpoints, since there are too many of them. Instead, we go over the breakpoints of the second half of $p_1$. For each such breakpoint  with key $k$, we check if the new polyline has a breakpoint with key $\lambda - k$. If so, denote its value by $w$, otherwise continue to the next breakpoint of $p_1$. These are points where we might have increased the value of $p_2$. We then query the new polyline with a \emph{value predecessor} query: this returns the breakpoint with the largest key s.t. its key is smaller than $\lambda - k$ and its value is at least $w$. 
If this breakpoint exists, and it is not the predecessor of the breakpoint at $\lambda - k$, then the values of the new polyline between its successor breakpoint and $\lambda - k$ should all be $w$ (i.e. we delete all breakpoints in this interval and set the successor's value to $w$). If it does not exist, then the values between zero and $\lambda - k$ should be $w$ (i.e. we delete all the previous breakpoints). 
This ensures that the resulting polyline is monotonically decreasing. 
See  Figure~\ref{figure of the second case in the construction of the first half of the polyline}.

\begin{figure}[h]
\begin{center}
\includegraphics[scale=0.6]{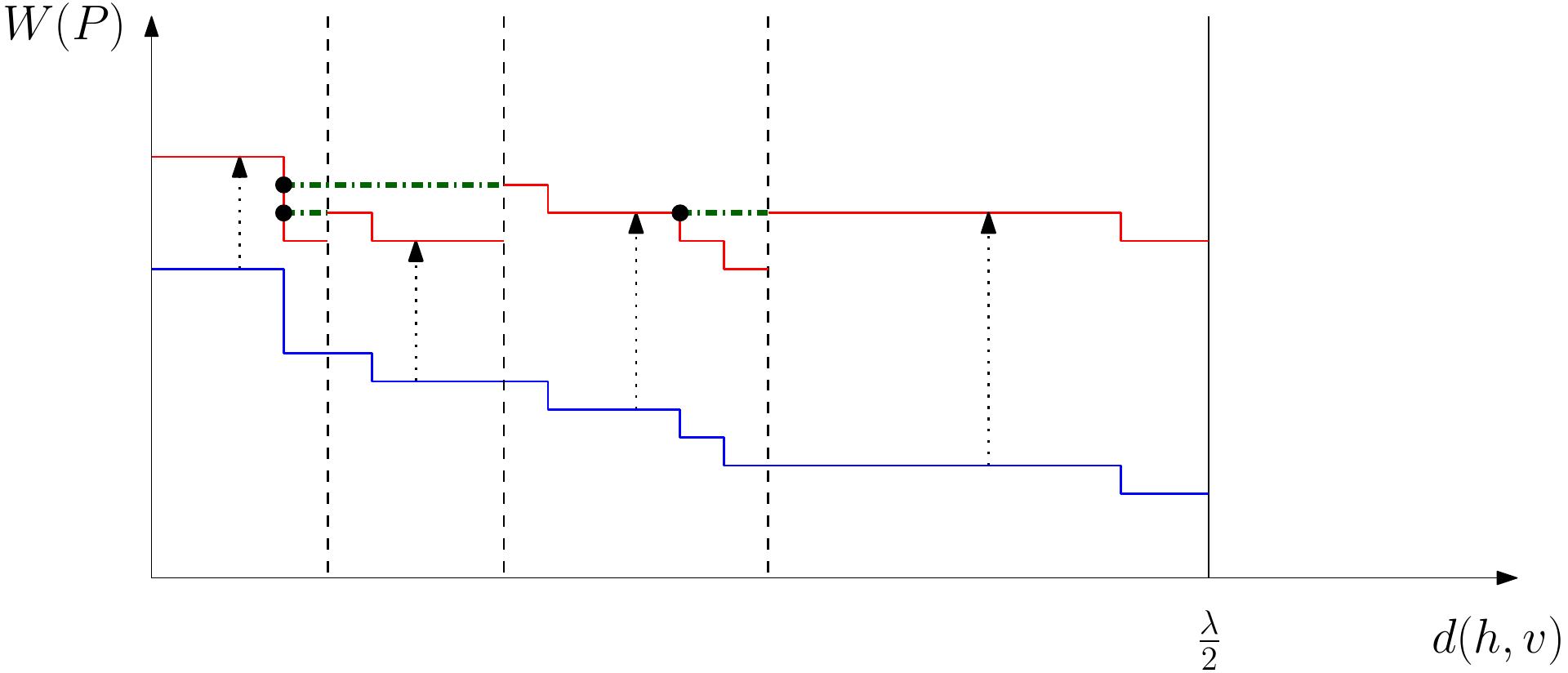}
\end{center}
\caption{Case II of constructing the first half of the polyline. The first half of $p_2$ is in blue. The vertical black dashed lines are the breakpoints of $p_1$. We increase the values of $p_{2}$ and obtain the red polyline, which is not monotone. To make it monotone, we delete all the breakpoints below the green intervals, which are found with value predecessor queries (the results of these queries are the bold points).\label{figure of the second case in the construction of the first half of the polyline}
}
\end{figure}

\vspace{0.04in} \noindent {\bf Merging cases I and II.}
We now need to build one polyline for the first half of the polyline, taking into account both cases. Let $p_a$ and $p_b$ denote the polylines we have constructed in cases I and II respectively (so the number of breakpoint in $p_a$ is at most $O(x)$, the number of breakpoints in $p_b$ is at most $O(y)$, and recall that $x\leq y$).

We now need to take the maximum of the values of $p_a$ and $p_b$, for each key. We do this by finding the intersection points of the two polylines. Notice that since both polylines are monotonically decreasing, these intersections can only occur at (i) the breakpoints of $p_a$, and (ii) at most one point between two consecutive breakpoints of $p_a$.

We iterate over $p_a$ and for each breakpoint, we check if the value of $p_b$ for the same key is between the values of this breakpoint and the predecessor breakpoint in $p_a$. If so, this is an intersection point. Then, we find the intersection points which are between breakpoints of $p_a$, by running a value predecessor query on $p_b$ for every breakpoint in $p_a$ except for the first. That is, for a breakpoint of $p_a$ with key $k$, denote the value of $p_a$ between it and its predecessor breakpoint by $v$. We find the point in $p_b$ with the largest key, s.t. its key is smaller than $k$, and its value is at least $v$. If such a point exists, and its key lies between the keys of the two consecutive breakpoints of $p_a$, and the value of the breakpoint of $p_b$ with this key is smaller than $v$, then it is an intersection point. Note that such a value predecessor query always returns a key corresponding to a breakpoint (or it returns NULL).

After such computation, we know which polyline gives us the best solution for every point between zero and $\frac{\lambda}{2}$, and where are the intersection points where this changes. We can now build the new polyline by doing insertions and deletions in $p_b$ according to the intersection points: For every interval of $p_b$ defined by a pair of consecutive intersection points, we check if the value of $p_a$ is larger than the value of $p_b$ in the interval, and if so, delete all the breakpoints of $p_b$ in the interval, and insert the relevant breakpoints from $p_a$. The number of intersection points is linear in the number of breakpoints of $p_a$, and so the total number of interval deletions and insertions is $O(x)$.

To conclude, the final polyline $p$ is obtained by concatenating the value computed for key zero, the polyline computed for the first half, and the polyline computed for the second half.

\subsection{The polyline data structure} We now specify the data structure for storing the polylines. The required interface is:
\begin{enumerate}
\item \label{op1} Split the polyline at some key.
\item \label{op2} Merge two polylines (s.t. all  keys in one polyline are smaller than all keys in the other).
\item \label{op3} Retrieve the value of the polyline for a certain key $d(h,v)$.
\item \label{op4}Return a sorted list of the breakpoints of the polyline.
\item \label{op5} Batched interval increase -- Given a list of disjoint intervals of the polyline, and a number for each interval, increase the values of the polyline in each interval by the appropriate number. Each interval is given by the keys of its endpoints.
\item \label{op6} Batched value predecessor -- Given a list of key-value pairs, $(k_i,v_i)$, find for each $k_i$, the maximal key $k_{i}'$, s.t. $k_{i}' < k_i$ and the value of the polyline at $k_{i}'$ is at least $v_i$, assuming that the intervals $(k_{i}',k_{i})$ are disjoint.
\item \label{op7} Batched interval insertions -- Given a list of pairs of consecutive breakpoints in the polyline, insert between each pair a list of breakpoints.
\item \label{op8} Batched interval deletions -- Given a list of disjoint intervals of the polyline, delete all the breakpoints inside the intervals.
\end{enumerate}
We now describe the data structure implementing the above interface. We represent a polyline by storing its breakpoints in an augmented 2-3 tree, where the data is stored in the leaves. Each node stores a key-value pair, and we maintain the following property: the key of each breakpoint is the sum of the keys of the corresponding leaf and all of its ancestors, and similarly for the values. In addition, we store in each node the maximal sum of keys and the maximal sum of values on a path from that node to a leaf in its subtree. We also store in each node the number of leaves in its subtree. Operations~\ref{op1} and~\ref{op2} use standard split and join procedures for 2-3 trees in logarithmic time. 
Operation~\ref{op3} runs a predecessor query and returns the value stored at the returned breakpoint in logarithmic time.
Operation~\ref{op4} is done by an inorder traversal of the tree. This takes $O(x)$ time, since we only iterate over the breakpoints of $p_{1}$. 
Operations~\ref{op1}-\ref{op4} are performed only a constant number of times per step, and so their total cost is $O(\log x + \log y + x)$. The next four operations are more costly, since they consists a batch of $O(x)$ operations. The input for the batched operations is given in sorted order (by keys).

\vspace{0.04in} \noindent {\bf Operation~\ref{op5} -- batched interval increase.}
Consider the following implementation for Operation \ref{op5}. We iterate over the intervals, and for each of them, we find its left endpoint, and traverse the path from the left endpoint, through the LCA, to the right endpoint.
The traversal is guided by the maximal key field stored in the current node. Using this field, we find the maximal key of a breakpoint stored in the node's subtree by adding the sum of all keys from the root to the current node, which is maintained in constant time after moving to a child or the parent.
While traversing the path from the left endpoint to the LCA (from the LCA to the right endpoint), we increase the value of every node hanging to the right (left) of this path. We also update the maximal value field in each node we reach (including the nodes on the path from the LCA to the root). Notice that if one of the endpoints of the interval is not in the structure, we need to insert it. We might also need to delete a breakpoint if it is a starting point of some interval and its new value is now equal to the value of its predecessor. This implementation would take time which is linear in the number of traversed nodes, plus the cost of insertions and deletions (whose number is linear in the number of intervals). Because the depth of a 2-3 tree of size $O(y)$ is $O(\log y)$, this comes up to $O(x \log y$). Such time complexity for each step would imply $O(n\log^{2}n)$ total time for the feasibility test.

We improve the running time by performing the operations on smaller trees. The operation therefore begins by splitting the tree into $O(x)$ smaller trees, each with $O(\frac{y}{x})$ leaves. This is done by recursively splitting the tree, first into two trees with $O(\frac{y}{2})$ leaves, then we split each of these trees into two trees with $O(\frac{y}{4})$ leaves, and so on, until we have trees of size $O(\frac{y}{x})$. We then increase the values in the relevant intervals using the small trees. For this, we scan the roots of the small trees, searching for the left endpoint of the first interval (by using the maximal key field stored in the root of each tree). Once we have found the left endpoint of the interval, we check if the right endpoint of the interval is in the same tree or not (again, using the maximal key). In the first case, the interval is contained in a single tree, and can be increased in this tree in time $O(\log(\frac{2y}{x}))$\footnote{The height of each small tree is actually bounded by $O(\log (\ceil{\frac{y}{x}}) = O(\log (\frac{2y}{x}))$.} using the procedure we have previously described. In the second case, the interval spans several trees, and so we need to do an interval increase in the two trees containing the endpoints of the interval, and additionally increase the value stored in the root of every tree that is entirely contained in the interval. We then continue to the next interval, and proceed in the same manner. 
Since the intervals are disjoint and we do at most two interval increases on small trees per interval, the total time for the increases in the small trees is $O(x \cdot \log (\frac{2y}{x}))$. Scanning the roots of the small trees adds $O(x)$ to the complexity, leading to 
 $O(x \cdot \log (\frac{2y}{x}) + x) = O(x \log (\frac{2y}{x}))$ overall for processing the small trees.

Before the operation terminates, we need to join the small trees to form one large tree. This is symmetric to splitting and analyzed with the same
calculation.
\begin{lemma}
\label{running time to obtain small trees lemma}
The time to obtain the small trees is $O(x \log (\frac{2y}{x}))$.
\end{lemma}
\begin{proof}
We start with a tree that has $y$ leaves, select the middle leaf by traversing from the root and using the number of leaves stored in every node of the tree, and then split the tree into two in $O(\log y)$ time. Then split the resulting two trees in $O(\log(\frac{y}{2}))$ time each, and so on. The cost of this process sums up to:
\begin{align*}
 \sum_{i=0}^{\log (\frac{x}{2})} 2^i \cdot \log (\frac{y}{2^i}) & = \sum_{i=0}^{\log (\frac{x}{2})} 2^i \cdot \log y - \sum_{i=0}^{\log (\frac{x}{2})} 2^i \cdot \log (2^i)  \\
& \leq x \log y - 2 \cdot (2^{ \log ( \frac{x}{2})} \cdot \log ( \frac{x}{2}) - 2^{ \log ( \frac{x}{2})}+1)  \\
& = x \log y  -  x \log (\frac{x}{2}) + x - 2 \\ 
&= x \log (\frac{2y}{x}) + x -2 \\
& = O(x \log (\frac{2y}{x})).
\end{align*}
\end{proof}

The cost of all joins required to patch the small trees together can be bounded by the same calculation as the cost of the splits made to obtain them, and so the operation takes $O(x \log (\frac{2y}{x}))$ time in total.

\vspace{0.04in} \noindent {\bf Operation~\ref{op6} -- batched value predecessor.}
Similarly to Operation~\ref{op5}, we start by describing a procedure that works on one tree, then apply the same trick of splitting the tree into small trees, and work on them.

We start each value predecessor query by first finding $k_i$ (or its predecessor, if there is no breakpoint at $k_i$). Then, we go up the path in the tree until the maximum value of a breakpoint corresponding to a leaf in a subtree hanging to the left of the path is at least $v_i$. We then keep going down to the rightmost child that has at least $v_i$ as the maximum value in its subtree. The overall runtime is $O(x \log y)$ since each value predecessor query takes $O(\log y)$ time.

We now apply the same splitting procedure as in Operation \ref{op5}, and obtain $O(x)$ trees with $O(\frac{y}{x})$ leaves each. We would like to iterate over the given keys and run value predecessor queries. 
In order for the time we spend scanning the roots of the trees to be bounded by $O(x)$, we need the intervals defined by the pairs $(k_i',k_i)$ to be disjoint. Recall that we had two uses for the batched value predecessor operation. The first was to prune the non-monotone polyline we have obtained by increasing intervals of $p_2$ (see Figure~\ref{figure of the second case in the construction of the first half of the polyline}). The pruning is done by deleting the intervals returned by this operation.
Consider two consecutive keys $k_1$ and $k_2$ ($k_1 \leq k_2$) for which we would like to find their value predecessors $k_1'$ and $k_2'$. If the value of the polyline at $k_1$  is greater or equal to the value  of the polyline at $k_2$, then certainly $k_1 \leq k_2'$, and the intervals are in fact disjoint. Else, the intervals might not be disjoint, but we know that $k_1' \geq k_2'$, and so the interval $(k_1',k_1)$ is contained in the interval $(k_2',k_2)$. Thus, deleting the second interval is enough, meaning that we do not need to run the value predecessor query for $k_1$ at all. Generalizing this observation, before running the value predecessor queries, we can remove entries from the input list
so that the values of the remaining entries are  monotonically decreasing, in $O(x)$ time. 
We can then answer value predecessor queries for this pruned list of keys using our small trees with $O(\frac{y}{x})$ leaves. We process the keys from right to left while simultaneously scanning the small trees (also from right to left). To process a pair $(k_{i},v_{i})$ we first find $k_{i}$, then we continue moving to the previous small tree as long as the maximal value stored at the root is smaller than $v_{i}$. Once we find the appropriate small subtree, we run a value predecessor query on it (similarly to what we described for one tree). 
Therefore, we search for the value predecessor of each key in one small tree in $O(\log(\frac{2y}{x}))$ time,
and additionally sweep through $O(x)$ roots, so the cost is $O(x \log (\frac{2y}{x}))$. We finish by merging the small trees back into one large tree (as in Operation \ref{op5}) which also takes $O(x \log (\frac{2y}{x}))$ time, and so this is also the time complexity of the entire operation.

The second use of this operation was to find the intersection points of two polylines, $p_a$ and $p_b$. In this case we want to find the intersection points that lie between two consecutive breakpoints of $p_a$, by running value predecessor queries on $p_b$. Thus, while performing a single value predecessor operation, we can stop the search if we have not found the value predecessor in the current interval defined by the pair of consecutive breakpoints of $p_a$. This guarantees that the time spent on scanning the roots of small the trees is $O(x)$, and so the running time of the operation is still $O(x \log (\frac{2y}{x}))$.

\vspace{0.04in} \noindent {\bf Operation~\ref{op7} -- batched interval insertions.}
Again, we split the tree into small trees with $O(\frac{y}{x})$ leaves. Additionally, we split the small trees s.t. every interval
that should be inserted falls between two small trees.
Since we have at most $x$ such extra splits, each of which takes $O(\log(\frac{2y}{x}))$ time, the total time for this is $O(x\log(\frac{2y}{x}))$.
Then, for each of the $O(x)$ breakpoints that are to be inserted, we create a small tree containing just a single leaf corresponding to that
breakpoint. Therefore, we have $O(x)$ small trees with $O(\frac{y}{x})$ leaves that should be now merged
to obtain one large tree.
The merging is done as in the previous operations, and the total cost can be bounded,
up to a constant factor, by essentially the same calculation as in Lemma \ref{running time to obtain small trees lemma}. Therefore, the overall time for the operation is $O(x \log(\frac{2y}{x}))$.

\vspace{0.04in} \noindent {\bf Operation~\ref{op8} -- batched interval deletions.}
Once again we split our tree into small trees. We then delete the relevant intervals from each of the small trees. Each interval is either contained in a single small tree (and is then deleted by performing one or two split operations and at most one join operation) or it spans several small trees (in which case we might also need to delete some trees entirely). This takes $O(x \log(\frac{2y}{x}))$ time in total.

\begin{theorem}
\label{nlogn weighted f.t. theorem}
The above implementation implies an $O(n \log n)$ weighted feasibility test.
\end{theorem}
\begin{proof}The running time of the algorithm is the sum of the time required for constructing the polylines of the subtrees rooted at each node of the tree.
Notice that the number of breakpoints in the polyline of a subtree is at most the number of nodes in this subtree. This is because every breakpoint corresponds to some node that was taken in the solution becoming too close to the root to be still included.

For any node that has one child (a degree one node), constructing its polyline is done with one query to the polyline of its child, and an increase of the keys of all but the first breakpoint (which can be done by increasing the key stored at the root, and decreasing the key stored at the first leaf). This takes $O(\log n)$ time and so the total time over all degree one nodes is $O(n \log n)$.
 
The time spent for each node that has two children (a degree two node) is bounded by the running time of the batched operations, which is $c\cdot x \log (\frac{2y}{x})$ for some constant $c$. We prove that this sums up to at most $2c \cdot n \log n$ by induction on the size of the  tree: Consider some binary tree with $n$ nodes. Denote the time spent for nodes of degree two while running the weighted feasibility test on the tree by $S$. If the root is of degree one, then by the induction hypothesis $S \leq 2c\cdot (n-1)\log(n-1)\leq 2c\cdot n \log n$. If the root is of degree two, denote the size of its smaller subtree by $x$, and the size of the larger by $y$ (so it holds that $n=x+y+1$, and $x \leq y$). Then, we have that
\begin{align*}
S = & \ 2c\cdot x \log x + 2c\cdot y \log y + c\cdot x \log (\frac{2y}{x})\\
 = & x\log x+y\log y+\frac{x}{2}+\frac{x}{2}\log y-\frac{x}{2}\log x \\
= & \frac{x}{2}\log x+y\log y+\frac{x}{2}+\frac{x}{2}\log y \\
= & \ 2c\cdot (\frac{x}{2}\log(2x)+(\frac{x}{2}+y)\log y) \\
\leq & \ 2c\cdot (\frac{x}{2}\log(x+y)+(\frac{x}{2}+y)\log (x+y)) \\
= & \ 2c\cdot ((x+y)\log(x+y)) \\
\leq & \ 2c\cdot n \log n.
\end{align*}
This concludes the induction and yields our $O(n \log n)$ weighted feasibility test.
\end{proof}

\bibliographystyle{abbrv}
\bibliography{arxiv}

\end{document}